\documentclass[11pt,a4paper]{article}
\usepackage[a4paper]{geometry}
\usepackage{amssymb,latexsym,amsmath,amsfonts,amsthm,enumerate}
\usepackage{graphicx}
\usepackage{epstopdf}
\usepackage{tikz}
\usepackage{pgflibraryshapes}
\usetikzlibrary{arrows,decorations.markings}
\usepackage{subfigure}
\usepackage{overpic}
\usepackage{fullpage}
\usepackage{comment,verbatim}
\usepackage{hyperref}
\usepackage{color}
\usepackage{mathrsfs}
\usepackage[title,titletoc]{appendix}
\usepackage{ulem}
\usepackage{enumitem}
\usepackage{authblk}

\def\arg {\mathop{\rm arg}\nolimits}
\def\Re {\mathop{\rm Re}\nolimits}

\def\Ai {{\rm Ai}}

\newtheorem{pro}{{Proposition}}
\newtheorem{remark}{{Remark}}
\newtheorem{lem}{{Lemma}}

\newtheorem{thm}{{Theorem}}

\numberwithin{equation}{section}

\begin{document}

\title{{Asymptotics of   Fredholm determinant solutions of the noncommutative Painlev\'e II equation
   }}

\author[1]{Jia-Hao Du}
\author[2]{Shuai-Xia Xu}
\author[1]{Yu-Qiu Zhao}
\affil[1]{Department of Mathematics, Sun Yat-sen University, Guangzhou, 510275, China}
\affil[2]{Institut Franco-Chinois de l'Energie Nucl\'{e}aire, Sun Yat-sen University, Guangzhou, 510275, China. E-mail: xushx3@mail.sysu.edu.cn}

\date{}

\maketitle

\noindent \hrule width 6.27in \vskip .3cm

\noindent {\bf{Abstract }} 
In this paper, we study the asymptotic behavior of a family of pole-free solutions to the noncommutative Painlev\'e II equation.  These particular solutions can be expressed in terms of  the Fredholm determinant of the matrix version of the classical Airy operator, which are analogous to the Hastings-McLeod solution and the Ablowitz-Segur solution of the classical Painlev\'e II equation. 
  Using the Riemann-Hilbert approach, 
we derive the asymptotics of the Fredholm determinant and the associated particular solutions  $\beta(\Vec{s})$ to the noncommutative Painlev\'e II equation in the regime $\Vec{s}=\left(s+\frac{\tau}{\sqrt{-s}},s-\frac{\tau}{\sqrt{-s}}\right)$ with $\tau\ge 0$  and $s\to-\infty$. The solutions depend on a two by two Hermitian matrix with eigenvalues in the interval $(-1,1)$.
The asymptotics are expressed in terms  of  one parameter family of special solutions of the classical Painlev\'e V equation. Furthermore, we derive the asymptotics, including the connection formulas,  for this one parameter family of solutions of the Painlev\'e V equation both as $ix\to -\infty$ and $x\to 0$.

\vskip .5cm

\noindent {\it{Keywords and phrases:}} 
Asymptotic analysis; Fredholm determinants; noncommutative Painlev\'e II equation; Painlev\'{e} V equation; Riemann-Hilbert problems; Deift-Zhou method. 

\vskip .3cm

\noindent \hrule width 6.27in\vskip 1.3cm

\tableofcontents

\section{Introduction and statement of results}
 
In random matrix theory, a fundamental problem is to determine the limiting distributions of the eigenvalues of large random matrices. 
Many studies have revealed the universality \cite{A,CK,CKV,Deift,DMVZ1,DMVZ2,forrester2010log} of these limiting distributions, among which the Tracy-Widom distribution functions play important roles. Consider Dyson's $\beta$-ensemble  with $\beta >0$ \cite{forrester2010log}, which
can be defined as a Coulomb gas of $N$ charged particles in the space of one dimension. The locations of these $N$ particles $\lambda _1<\lambda_2<...<\lambda_N$ have the following probability density function
\begin{equation}\label{eq:prob}
    p(\lambda_1,...,\lambda_N)d\lambda_1\cdots d\lambda_N=\frac{1}{Z_N}\prod_{1\le j<k\le N}|\lambda_j-\lambda_k|^{\beta}\mathrm{e}^{-\beta\sum_{j=1}^N V(\lambda_{j})}d\lambda_1\cdots d\lambda_N,
\end{equation}
where $Z_N$ is a normalization constant, and  the external field $V(\lambda)$ is assumed to be Gaussian, namely, $V(\lambda)=\frac{\lambda^2}{2}$. For $\beta=1,~2,~4$, the $\beta$-ensembles are known as Gaussian orthogonal, Gaussian unitary and Gaussian symplectic ensembles, respectively. 
For $ \beta >0$, the  Tracy-Widom distribution functions $F_{\beta}(t)$  describe the edge fluctuations of these eigenvalues in the large $N$ limit  \cite{forrester2010log} 
\begin{equation}\label{eq:TWd}
    F_{\beta}(t)=\lim\limits_{N\to\infty}\mathrm{Prob}\left((\lambda_{N}-\sqrt{2N})2^{\frac{1}{2}}N^{\frac{1}{6}}\le t\right).
\end{equation}

For $\beta=1,~2,~4$,  the distribution functions $F_{\beta}(t)$ become the classical Tracy-Widom distributions \cite{TW}. They can be represented  explicitly both as the Fredholm determinant of the Airy kernel  and integral of the Hastings-McLeod solution of the  Painlev\'e II equation.
Taking $\beta=2$ for example,  we denote by $\mathbb{K}_{\mathrm{Ai},t}$ the trace operator acting on  $L^2((t,\infty),\mathbb{R})$ with the Airy kernel
\begin{equation}\label{eq:kAiry}
    K_{Airy}(x,y)=\frac{\Ai(x)\Ai'(y)-\Ai'(x)\Ai(y)}{x-y},
\end{equation}
where $\Ai(z)$ is the Airy function \cite{O},
\begin{equation}
    \Ai(z)=\frac{1}{2\pi}\int_{-\infty}^{+\infty}\mathrm{e}^{-\frac{i}{3}s^3-izs}ds.
\end{equation}
The Tracy-Widom distribution $F_2(t)$ can be expressed by the Fredholm determinant of the Airy kernel
\begin{equation}\label{eq:def F2}
    F_2(t)=\det (Id-\mathbb{K}_{\mathrm{Ai},t}).
\end{equation}
It is remarkable that the Fredholm determinant admits the following integral representation; cf.
\cite{TW},
\begin{equation}\label{eq:int}
   \det (Id-\mathbb{K}_{\mathrm{Ai},t})=\mathrm{exp}\left(-\int_t^{\infty}(s-t)u^2(s)ds\right),
\end{equation}
where $u(t)$ is the Hastings-McLeod solution \cite{hastings} of the second Painlev\'e equation
\begin{equation}\label{eq:PII}
    u''=tu+2u^3,
\end{equation}
determined by the boundary condition
\begin{equation}\label{eq:HMSofPII}
\begin{split}
    u(t)\sim\begin{cases}
        \Ai(t),&t\to+\infty,\\
        \sqrt{-\frac{t}{2}},&t\to-\infty.\end{cases}
    \end{split}
\end{equation}
These representations allow one to derive the following asymptotic expansion as $t\to-\infty$ \cite{ TW}
\begin{equation}\label{eq:gam=1}
    \frac{d}{dt}\ln F_2(t)=\frac{t^2}{4}-\frac{1}{8t}+\mathcal{O}((-t)^{-4}),~~~t\to-\infty.
\end{equation}
Furthermore, we can obtain the asymptotic expansion of  
$F_2(t)$ as $t\to-\infty$ and the explicit expression of the constant term \cite{DIK, E}.

Now consider a thinned process for the Gaussian unitary ensemble, which is obtained by removing each eigenvalue independently with the probability 
$1-\gamma$, $\gamma\in[0,1]$; see \cite{bo09, bo04, bo06}. Then, the distribution function for the largest remaining eigenvalue $\lambda_{max}$  in this thinned process is given by 

\begin{equation}\label{eq:gTWd}
    F_{2}(t;\gamma)=\lim\limits_{N\to\infty}\mathrm{Prob}\left((\lambda_{max}-\sqrt{2N})2^{\frac{1}{2}}N^{\frac{1}{6}}\le t\right).
\end{equation}
Similar to $F_2(t)$, the above distribution admits both   Fredholm determinant representation and integral expression \cite{BCI}:
\begin{equation}\label{eq:gKAiry}
    F_2(t;\gamma)=\det (Id-\gamma \mathbb{K}_{\mathrm{Ai},t})=\mathrm{exp}\left(-\int_t^{\infty}(s-t)u^2(s;\gamma)ds\right),
\end{equation}
whereas the Hastings-McLeod solution of the Painlev\'e II equation is replaced by the Ablowitz-Segur solution with the following asymptotic behavior \cite{Ab,FIKY,O}:
\begin{equation}\label{eq:ABSofPII}
     u(t;\gamma)\sim\begin{cases}
        \sqrt{\gamma}\Ai(t),&t\to+\infty,\\
        (-t)^{-\frac{1}{4}}\sqrt{2\rho}\cos \left(\frac{2}{3}(-t)^{\frac{3}{2}}-\rho\ln (8(-t)^{\frac{3}{2}})+\frac{\pi}{4}-\arg\Gamma(-i\rho)\right),&t\to -\infty,
     \end{cases}
\end{equation}
where $\rho=-\frac{1}{2\pi}\ln (1-\gamma)$ with $\gamma\in [0,1)$.
 ~Substituting  \eqref{eq:ABSofPII} into the expression of $F_2(t;\gamma)$ implies that as $t\to -\infty$ \cite{baik, BCI, Bothner19}
\begin{equation}\label{eq:logTWf}
        \frac{d}{dt} \ln F_2(t;\gamma)=2\rho\sqrt{-t}+\frac{3\rho^2}{2t}-\frac{\rho}{2t}\sin(2\phi(t,\rho))+\mathcal{O}((-t)^{-\frac{7}{4}}),~\gamma\in[0,1),
    \end{equation}    where \begin{equation}
        \phi(t,\rho)=\frac{2}{3}(-t)^{\frac{3}{2}}-\rho\ln (8(-t)^{\frac{3}{2}})+\frac{\pi}{4}-\arg\Gamma\left(-i\rho\right).    \end{equation}

In \cite{ber12}, Bertola and Cafasso  
studied a matrix version of the Airy-convolution operator $Ai_{\Vec{s}}$ acting on $L^2(\mathbb{R}_+,\mathbb{C}^r)$:
\begin{equation}\label{eq:Mat KAiry}
    \begin{split}
         (Ai_{\Vec{s}}\Vec{\boldsymbol{f}})(x):=& \int_{R_+}\textbf{Ai}(x+y;\Vec{s})\Vec{\boldsymbol{f}}(y)dy,\\
    \textbf{Ai}(x;\Vec{s}):=&
    \int_{\gamma_+}\mathrm{e}^{\theta(\mu)}C\mathrm{e}^{\theta(\mu)}\mathrm{e}^{ix\mu}\frac{d\mu}{2\pi}=\left(c_{jk}\Ai(x+s_j+s_k)\right)_{j,k=1}^r,\\
    \theta(\mu):=&\frac{i\mu^3}{6}I_r+\mathrm{diag}(s_1,s_2,...,s_r)i\mu ,
    \end{split}
\end{equation}
where $\Vec{\boldsymbol{f}}=(f_1,...,f_r)^T$, $\Vec{s}=(s_1,...,s_r)$, the matrix $I_r$ is a $r \times r$ identity matrix, $C=\left(c_{jk}\right)_{j,k=1}^r$ is a $r\times r$ Hermitian matrix and the contour $\gamma_+$ is a contour contained in the upper half plane and extending to infinity along the directions $\arg z=\frac{\pi}{2}\pm\frac{\pi}{3}$. For $r=2$, we can choose a unitary matrix $U$ such that 
\begin{equation}\label{eq:def C}
    C=U\begin{pmatrix}
        \lambda_1 & 0\\
        0 & \lambda_2
    \end{pmatrix}U^{-1},
\end{equation}
where $\lambda_1$ and $\lambda_2$ are the eigenvalues of the Hermitian matrix $C$ with $\lambda_1, \lambda_2 \in[-1,1]$. 
The unitary matrix can be written as 
    \begin{equation}\label{eq:def U}
        U=\begin{pmatrix}
        \omega & \sqrt{1-|\omega|^2}\\
        -\sqrt{1-|\omega|^2} & \overline{\omega}
         \end{pmatrix},       
    \end{equation}
    where each column corresponds to the eigenvector of $C$. The parameter $\omega\in\mathbb{C}$, $|\omega|\le 1$ and $\overline{\omega}$ denotes the complex conjugation of $\omega$. 

In \cite{ber12}, Bertola and Cafasso
showed that the square of the matrix Airy-convolution operator \eqref{eq:Mat KAiry} is totally positive and thus defines a determinantal point process, if $C$ is a Hermitian matrix with eigenvalues in $[-1,1]$. The corresponding Fredholm determinant then gives the gap probability for this process. Furthermore, they showed that the Fredholm determinant \eqref{eq:int of FD} is related to a family of particular solutions for the noncommutative Painlev\'e II equation 
\begin{equation}\label{eq:int of FD}
    \det(Id-Ai_{\Vec{s}}^2)=\mathrm{exp}\left(-4\int_s^{\infty}(t-s)Tr(\beta^2(t+\Vec{\eta}))dt\right),
\end{equation}
where 
$\Vec{\eta}=(\eta_1,...,\eta_r):=(s_1-s,...,s_r-s),~s=\frac{1}{r}\sum_{j=1}^{r}s_j$. Here, the  $r\times r$ matrix-valued functions $\beta(\Vec{s})=\beta(\Vec{s};C)$ represents a family of solutions to the following noncommutative Painlev\'e II (PII for short) equation \cite{ber12,ber18}:
\begin{equation}\label{eq:NC PII}    D^{2}\beta(\Vec{s})=4\left(\textbf{s}\beta(\Vec{s})+\beta(\Vec{s})\textbf{s}\right)+8\beta(\Vec{s})^{3},~ \textbf{s}:=\mathrm{diag}(s_1,...,s_r), ~D:=\sum_{i=1}^{r}\frac{\partial}{\partial s_i},
\end{equation}
which first appeared in \cite{rr}. If $s=\frac{1}{r}\sum_{j=1}^{r}s_j\to +\infty $ 
and $|s_j-s|\le m,~j=1,2,...,r$ are kept fixed, then $\beta(\Vec{s})$ satisfy the asymptotic behavior
\begin{equation}\label{eq:asy NC PII at +infty}
    (\beta)_{kl}(\Vec{s})\sim -c_{kl}\Ai(s_k+s_l),~~s\to+\infty.
\end{equation}
Additionally, if $C$ is Hermitian then $\beta(\Vec{s})$ is pole-free for all $\Vec{s}\in\mathbb{R}^r$ if and only if the eigenvalues of $C$ are in $[-1,1]$; see \cite[Theorem 5.3]{ber12}. When $r=1$, the determiant in \eqref{eq:int of FD} reduces to the celebrated Tracy-Widom distribution \eqref{eq:int}. For  general $r\geq 2$, the determinant in \eqref{eq:int of FD} is the noncommutative analogue of the Tracy-Widom distribution and the corresponding special solutions $\beta(\Vec{s})$ to the noncommutative PII equation are the noncommutative  analogue of the classical Hastings-McLeod solution and Ablowitz-Segur solution.  
We denote
\begin{equation}\label{eq:def delta(s)}
    D{H}_{II}(\Vec{s})=-2i\beta(\Vec{s})^2.
\end{equation}
If $r=1$, ${H}_{II}(\Vec{s})$ is the Hamiltonian for the classical PII equation \eqref{eq:PII}. For  general $r\geq 2$, we can regard ${H}_{II}(\Vec{s})$ as the Hamiltonian for the noncommutative PII equation $\beta(\Vec{s})$. Since $C$ is Hermitian, we then have the following symmetry relations for $\beta(\Vec{s})$ and ${H}_{II}(\Vec{s})$; cf. \cite{ber12},
\begin{equation}
    \beta(\Vec{s})=\beta(\Vec{s})^*,~{H}_{II}(\Vec{s})=-{H}_{II}(\Vec{s})^*,
\end{equation}
where $\beta(\Vec{s})^*$ denotes the Hermitian conjugate of $\beta(\Vec{s})$. 

It is remarkable that  the Tracy-Widom distribution functions  $F_{\beta}(t)$ defined in \eqref{eq:TWd} with the even values of parameter $\beta$ admit integral expression in terms of the  eigenvalues of a particular solution of the noncommutative PII equation; see \cite[Eq. (3.29)]{Ru16}, \cite{Ru15} and  \cite[ Eq. (1.21)]{IP}.  We mention that the eigenvalues of the noncommutative PII equation satisfy the corresponding second Calogero-Painlev\'e system as shown in \cite{ber18}. Recently, Its and Prokhorov \cite{IP}
studied the  tail  asymptotics of  the  $\beta=6$ Tracy-Widom distribution and the associated  solution of the noncommutative PII equation
by using the Riemann-Hilbert problem representation for the noncommutative PII equation discovered in \cite{ber12,ber18}. It should be pointed out that the Stokes multipliers of the Riemann-Hilbert problem considered in \cite{IP} differ from those in our case.

In the present work, we study the asymptotics of the Fredholm determinant \eqref{eq:int of FD} and the related special solutions of the noncommutative PII equation \eqref{eq:NC PII} for the first nontrivial case $r=2$. In the regime $\Vec{s}=\left(s+\frac{\tau}{\sqrt{-s}},s-\frac{\tau}{\sqrt{-s}}\right)$ with $\tau\ge0$ and $s\to-\infty$, we show that  the asymptotics of the Fredholm determinant \eqref{eq:int of FD} and the associated solutions to the noncommutative PII equation can be expressed in terms  of a family of special solutions of the classical Painlev\'e V   (PV for short)  equation. Furthermore, we derive the asymptotics of  this one parameter family of solutions of the PV equation both as $ix\to -\infty$ and $x\to 0$. For general $r>2$, it is expected that the asymptotics of the Fredholm determinant \eqref{eq:int of FD} would be described by the noncommutative PV equation; see Remark \ref{re:NC PV} below for the discussion. We will leave this problem for future investigation. It is noted that the  asymptotics of the noncommutative PII equation is also considered recently in \cite{LYZH} under a different asymptotic regime than the one addressed in our study.

\subsection{Asymptotics of  one parameter family of  solutions of the  PV equation}
To state our main results, we need to introduce the PV equation. Let $u(x), v(x)$ and $ y(x)$ satisfy the following nonlinear ODEs
\begin{equation}\label{eq:PV for u}
    x\frac{du}{dx}=-xu-2v(u-1)^2-\frac{1}{2}(u-1)[(\theta_0-\theta_1+\theta_{\infty})u-(3\theta_0+\theta_1+\theta_{\infty})],
\end{equation}
\begin{equation}\label{eq:PV for v}
      x\frac{dv}{dx}=uv[v+\frac{1}{2}(\theta_0-\theta_1+\theta_{\infty})]-\frac{1}{u}(v+\theta_0)[v+\frac{1}{2}(\theta_0+\theta_1+\theta_{\infty})],
\end{equation}
\begin{equation}\label{eq:PV for y}
    x\frac{d \ln y}{dx}=-2v-\theta_0+u[v+\frac{1}{2}(\theta_0-\theta_1+\theta_{\infty})]+\frac{1}{u}[v+\frac{1}{2}(\theta_0+\theta_1+\theta_{\infty})].
\end{equation}
These equations are the compatibility condition, namely $\Phi_{\lambda x}=\Phi_{x \lambda}$, of the following Lax pair for the $2 \times 2$ matrix-valued function $\Phi(\lambda,x)$ \cite{FIKY,FMZ,XZ}:
\begin{equation}\label{eq:lax pair L}
    \frac{\partial \Phi}{\partial \lambda}=L\Phi ,~~~ L(\lambda,x)=-\frac{x}{2}\sigma_3+\frac{L_1(x)}{\lambda}+\frac{L_2(x)}{\lambda-1}, 
\end{equation}
and
\begin{equation}\label{eq:lax pair U}
      \frac{\partial \Phi}{\partial x}=U\Phi
      ,~~~U(\lambda,x)=-\frac{\lambda}{2}\sigma_3+B_0(x),
\end{equation}
with
\begin{equation}\label{eq:def L_1}
    L_1(x)=\begin{pmatrix}
        v+\frac{\theta_0}{2} & -y(v+\theta_0)\\
        \frac{v}{y} & -(v+\frac{\theta_0}{2})
    \end{pmatrix},
\end{equation}
\begin{equation}\label{eq:def L_2}
    L_2(x)=\begin{pmatrix}
        -v-\frac{1}{2}(\theta_0+\theta_{\infty}) & yu(v+\frac{1}{2}(\theta_0-\theta_1+\theta_{\infty}))\\
        -\frac{1}{yu}(v+\frac{1}{2}(\theta_0+\theta_1+\theta_{\infty})) & v+\frac{1}{2}(\theta_0+\theta_{\infty})
    \end{pmatrix},
\end{equation}
\begin{equation}
    B_0(x)=\begin{pmatrix}
        0 & \frac{y}{x}[v+\theta_0-u(v+\frac{1}{2}(\theta_0-\theta_1+\theta_{\infty}))]\\
        -\frac{1}{xy}[v-\frac{1}{u}(v+\frac{1}{2}(\theta_0+\theta_1+\theta_{\infty}))] & 0
    \end{pmatrix}.
\end{equation}

Using equation \eqref{eq:PV for v} to eliminate $v(x)$ in \eqref{eq:PV for u}, it follows that $u(x)$ satisfies the PV equation

\begin{multline}\label{eq:PV}
    \frac{d^{2} u}{d x^2}=\left(\frac{1}{2u}+\frac{1}{u-1}\right)\left(\frac{du}{dx}\right)^2-\frac{1}{x}\frac{du}{dx}+\frac{(u-1)^2}{8x^2}\left((\theta_0-\theta_1+\theta_{\infty})^2u-\frac{(\theta_0-\theta_1-\theta_{\infty})^2}{u}\right)\\
    -\frac{(1-\theta_0-\theta_1)u}{x}-\frac{u(u+1)}{2(u-1)}.
    \end{multline}
Moreover, we have the  Hamiltonian for the PV equation \cite{JM}
\begin{equation}\label{eq:def for HV}
    H_V(x)=\frac{1}{x}\left(v-\frac{1}{u}(v+\frac{\theta_0+\theta_1+\theta_{\infty}}{2})\right)\left(v+\theta_0-u(v+\frac{\theta_0-\theta_1+\theta_{\infty}}{2})\right)-v-\frac{\theta_0+\theta_{\infty}}{2}.
\end{equation}
For our application, we have 
\begin{equation}\label{eq:para condition}
    \theta_0=-\theta_1=-2\alpha\in i\mathbb{R},~ \theta_{\infty}=0.
\end{equation}
For later use, we define the following function $\hat{u}(x)$, which is related to the PV equation
\begin{equation}\label{eq:def hatu}
    \hat{u}(x)=a_0b_1y\left(\frac{v-2\alpha}{v}-u\right), 
    \end{equation}
where $a_0(x),~b_1(x)$ are defined in \eqref{eq:analyticfac}.

We prove the following existence and asymptotics for one parameter family of solutions of the PV equation. 
\begin{thm}\label{thm:2}
    Let $\alpha\in i\mathbb{R}$, $\omega\in\mathbb{C}$ and $|\omega|\le 1$. There exists a $\omega$-family of pole-free solutions on the imaginary axis to the Painlev\'e V system \eqref{eq:PV for u}-\eqref{eq:PV for y}, which satisfies the following asymptotic behavior
\begin{align}
 &u(x)=\mathrm{e}^{-x+4\mu\ln (-ix)}\frac{\Gamma(1+\alpha-\mu)}{\Gamma(1+\mu-\alpha)}\frac{\Gamma(1-\alpha-\mu)}{\Gamma(1+\alpha+\mu)}+\mathcal{O}(x^{-1}),~ix\to-\infty,\label{eq:asyofu}\\
&v(x)=\alpha+\mu+\frac{2(\mu^2-\alpha^2)}{x}(2\sin^2{d(x)}-1)+\mathcal{O}(x^{-2}),~ix\to-\infty,\label{eq:asyofv}\\
&y(x)=\frac{\sqrt{1-|\omega|^2}}{\overline{\omega}}\mathrm{e}^{\pi i\alpha}\mathrm{e}^{-2\mu\ln(-ix)}\frac{\Gamma(1+\alpha+\mu)}{\Gamma(1+\alpha-\mu)}+\mathcal{O}(x^{-1}),~\omega\neq 0,~ix\to-\infty,\label{eq:asyofy}
\end{align}
\begin{equation}\label{eq:asyofhu}
    \begin{aligned}
        \hat{u}(x)=\mathrm{e}^{-2\alpha\ln(-ix)+2i\arg\Gamma(2\alpha)+i\arg\Gamma(\mu-\alpha)-i\arg\Gamma(\mu+\alpha)}\Big(&\cos d(x)-\frac{\mu}{\alpha}i\sin d(x)\Big)\\&+\mathcal{O}(x^{-1}),~ix\to-\infty,
    \end{aligned}
\end{equation}
\begin{align}
  &u(x)=1+\mathcal{O}(x),~x\to0,\label{eq:asyofu0} \\ 
     &v(x)= 2(1-|\omega|^2)\alpha+\mathcal{O}(x),~x\to0,\label{eq:asyofv0}\\
     &y(x)=\frac{\sqrt{1-|\omega|^2}}{\overline{\omega}}+\mathcal{O}(x),~\omega\neq 0,~ x\to0,\label{eq:asy y0}\\
     &\hat{u}(x)=1+\mathcal{O}(x),~x\to0,\label{eq:asy hatu 0}
\end{align}
    where 
\begin{equation}\label{eq:defofd}
        d(x)=-\frac{ix}{2}+2\mu i\ln(-ix)+\arg{\Gamma(\mu-\alpha)}-\arg{\Gamma(-\alpha-\mu)},
    \end{equation}
    and the connection formula between $\mu$ and $\omega$ is given by
    \begin{equation}\label{eq:connection formula}
        \mu=\frac{1}{2\pi i}\ln\left((1-|\omega|^2)\mathrm{e}^{2\pi i\alpha}+|\omega|^2\mathrm{e}^{-2\pi i\alpha}\right).
    \end{equation}
Moreover, the associated Hamiltonian $H_V(x)$ is pole-free for $x\in i\mathbb{R}$ and has the following asymptotic behavior
\begin{equation}\label{eq:asyofHV}
    H_V(x)=-\mu-\frac{2}{x}(\alpha^2-\mu^2)+\mathcal{O}(x^{-2}),~ix\to-\infty,
\end{equation}
and 
\begin{equation}
     H_V(x)=(2|\omega|^2-1)\alpha+\mathcal{O}(x),~x\to0. 
\end{equation}
Particularly, the functions $v(x)$ and $H_V(x)$ are purely imaginary on $i\mathbb{R}$.
\end{thm}

\begin{remark}
   We have the symmetry relations for this family of solutions to the PV equation
   $$u(-x)=u(x)^{-1}, ~v(-x)=-v(x), ~y(-x)=\mathrm{e}^xy(x)u(x);$$ 
 see  Remark \ref{re:symmetry -x}. Thus, 
   the asymptotic expansions of the solutions as $ix\to+\infty$ can be obtained using these symmetry relations and Theorem \ref{thm:2}.
\end{remark}
\begin{remark}
    If $\alpha=0$, we have $\hat{u}(x)=|\omega|^2\mathrm{e}^{\frac{x}{2}}+(1-|\omega|^2){e}^{-\frac{x}{2}}$; see Remark \ref{re:C diagonal}. This result is consistent with the asymptotic expansion  \eqref{eq:asyofhu} when taking the limit $\alpha\to0$.
\end{remark}
\subsection{Asymptotics of the Fredholm determinant and the noncommutative PII equation}\label{sec: results}

We will derive the asymptotics of the Fredholm determinant  \eqref{eq:int of FD}  and the associated particular solutions to the noncommutative PII equation  \eqref{eq:NC PII} 
as the variables tend to negative infinity at certain related speed. The asymptotics are expressed in terms  of  the one parameter family of special solutions of the classical PV equation with the properties given in Theorem \ref{thm:2}.  The results are sated in the following theorems.
\begin{thm}\label{thm:4}
 Let $Ai_{\Vec{s}}$ denote the matrix version of the Airy-convolution operator acting on $L^2(\mathbb{R}_+,\mathbb{C}^2)$ as defined in \eqref{eq:Mat KAiry}, where $\Vec{s}=\left(s+\frac{\tau}{\sqrt{-s}},s-\frac{\tau}{\sqrt{-s}}\right)$ with $\tau\ge 0$ and $s<0$. The eigenvalues of the Hermitian matrix $C$ in \eqref{eq:def C}, namely  $\lambda_1$ and $ \lambda_2$,  are in $(-1,1)$, and the eigenvectors of $C$ involve a parameter $\omega$, $\omega\in\mathbb{C}$ and $|\omega|\le 1$ as given in \eqref{eq:def U}. Then, as $s\to -\infty$, we have
\begin{multline}\label{eq:asy of FD}
   \partial_s\ln\det(Id-Ai^2_{\Vec{s}})={4\sqrt{2}}(\rho_1+\rho_2)\sqrt{-s}+\frac{3}{2s}(\rho_1^2+\rho_2^2)-\frac{1}{2s}\Bigg(K(\tau)\left(\rho_1\cos(d_1(s,\tau))+\rho_2\cos(d_2(s,\tau))\right)\\
\\+2(K(\tau)-1)\sqrt{\rho_1\rho_2}\cos\left(\frac{d_1(s,\tau)+d_2(s,\tau)}{2}\right)
-8\sqrt{2}\tau iH_V(4\sqrt{2}\tau i)\Bigg)+\mathcal{O}\left((-s)^{-\frac{7}{4}}\right),
\end{multline}
where 
\begin{equation}\label{eq:def d_1}
         d_j(s,\tau)=\frac{8}{3}\sqrt{2}(-s)^{\frac{3}{2}}-{2\rho_j}\ln(2^{\frac{9}{2}}(-s)^{\frac{3}{2}})-2\arg\Gamma\left(-i{\rho_j}\right)+2(-1)^j\arg\hat{u}(4\sqrt{2}\tau i),
\end{equation}
\begin{equation}\label{eq:def K}
    K(\tau)= 1+\frac{\sqrt{2}\tau i}{\alpha^2}\left(H_V(4\sqrt{2}\tau i)-\alpha+v(4\sqrt{2}\tau i)\right),~\alpha\neq0,~K(\tau)=1,~\alpha=0,
\end{equation}
and
\begin{equation}\label{eq:defv1,v2}
    \alpha=\frac{1}{4\pi i}\ln\left(\frac{1-\lambda_1^2}{1-\lambda_2^2}\right),~\rho_j=-\frac{1}{2\pi}\ln(1-\lambda_j^2),~~~j=1,2.
\end{equation}
Here  $v(x)$ is the solution of the PV system \eqref{eq:PV for u}-\eqref{eq:PV for y}  and $H_V(x)$ is  the associated  Hamiltonian with the properties specified in Theorem \ref{thm:2} and $\hat{u}(x)$  is defined by the relation \eqref{eq:def hatu}.
\end{thm}
    \begin{remark}\label{re:C2 diag}
        If $C^2$ is a diagonal matrix, then from Remarks \ref{re:C diagonal} and \ref{re:deg solution of S} below, we have
        \begin{equation}
         v(4\sqrt{2}\tau i)=0, ~ \hat{u}(4\sqrt{2}\tau i)=\mathrm{e}^{2\sqrt{2}\tau i},~ H_V(4\sqrt{2}\tau i)=\alpha, ~K(\tau)=1. 
        \end{equation}
          Thus, the asymptotic expansion \eqref{eq:asy of FD}
        becomes the following simpler form
        \begin{multline}\label{eq:asy FD in C is diag}
   \partial_s\ln\det(Id-Ai^2_{\Vec{s}})={4\sqrt{2}}(\rho_1+\rho_2)\sqrt{-s}+\frac{3}{2s}(\rho_1^2+\rho_2^2)-\frac{1}{2s}\Big(\rho_1\cos(d_1(s,\tau))\\
+\rho_2\cos(d_2(s,\tau)))-8\sqrt{2}\alpha\tau i\Big)+\mathcal{O}\left((-s)^{-\frac{7}{4}}\right).
\end{multline}
where
        \begin{equation}
                    d_j(s,\tau)=\frac{8}{3}\sqrt{2}(-s)^{\frac{3}{2}}-{2\rho_j}\ln(2^{\frac{9}{2}}(-s)^{\frac{3}{2}})-2\arg\Gamma\left(-i{\rho_j}\right)+(-1)^j4\sqrt{2}\tau i,~j=1,2.
        \end{equation}
          On the other hand, from \eqref{eq:Mat KAiry} and the fact that $C^2$ is a diagonal matrix, we have
             \begin{equation}\label{eq:deg result}
        \ln\det(Id-Ai^2_{\Vec{s}})= \ln F_2(s_1;\lambda_1^2)+\ln F_2(s_2;\lambda_2^2),
        \end{equation}
        where $\Vec{s}=(s_1,s_2)$ and $F_2(t;\gamma)$ is defined in \eqref{eq:gKAiry}. Therefore,  we  see that the result \eqref{eq:asy FD in C is diag} is consistent with \eqref{eq:deg result} and \eqref{eq:logTWf}. 
\end{remark}

\begin{remark}
    If $\Vec{s}=(s,s)$, i.e., $\tau=0$, the logarithmic of the Fredholm determinant of the matrix Airy kernel will also degenerate into the sum of two Airy kernel determinants, as in  \eqref{eq:deg result}. In this case, we have $\hat{u}(0)=1$ and $K(0)=1$ from Theorem \ref{thm:2}. Substituting this into \eqref{eq:asy of FD}, we obtain \eqref{eq:asy of FD} with $\tau=0$ therein. This is also consistent with \eqref{eq:deg result} and \eqref{eq:logTWf}.
\end{remark}
\begin{remark}
    If $\Vec{s}\in\mathbb{R}^r$  for $r>2$, we expect that
    the asymptotic expansions of $\ln\det(Id-Ai^2_{\Vec{s}})$ as $s_j\to-\infty$ at certain related speed  would be expressed in terms of particular solution of the noncommutative PV equation; see Remark \ref{re:NC PV} for more discussion. 
\end{remark}
\begin{remark}
   
    Note that the classical Airy kernel determinants $ F_2(t;\gamma)$ defined in \eqref{eq:gKAiry} exhibit different types of asymptotic behaviors for $\gamma=1$, as compared with $\gamma\in(0,1)$; cf. \eqref{eq:gam=1} and \eqref{eq:logTWf}.
   It is expected that  the determinant $\det(Id-Ai^2_{\Vec{s}})$ would have  asymptotic behavior different from that described in Theorem \ref{thm:4} if one of the eigenvalues of $C$ in  \eqref{eq:Mat KAiry} is equal to $\pm 1$ and the other one is  within the interval $(-1,1)$. 
\end{remark}
\begin{thm}\label{thm5}
    Under the assumptions of Theorem \ref{thm:4}, let $\beta(\Vec{s})$ be the solution of  the noncommutative Painlev\'e II equation \eqref{eq:NC PII} satisfying the asymptotic behavior \eqref{eq:asy NC PII at +infty} as $s\to+\infty$, then $\beta(\Vec{s})$ has the following asymptotic behavior as $s\to-\infty$,
       \begin{align}\label{eq:asy NC PII}
      (\beta(\Vec{s}))_{11}&=2^{\frac{1}{4}}\left(\sqrt{{\rho_2}}\frac{v(4\sqrt{2}\tau i)}{2\alpha}\sin\psi_2(s,\tau)-\sqrt{{\rho_1}}\frac{2\alpha-v(4\sqrt{2}\tau i)}{2\alpha}\sin\varphi_1(s,\tau)\right)(-s)^{-\frac{1}{4}}+\mathcal{O}((-s)^{-1}),\\
       (\beta(\Vec{s}))_{22}&=2^{\frac{1}{4}}\left(\sqrt{{\rho_2}}\frac{2\alpha-v(4\sqrt{2}\tau i)}{2\alpha}\sin\varphi_2(s,\tau)-\sqrt{{\rho_1}}\frac{v(4\sqrt{2}\tau i)}{2\alpha}\sin\psi_1(s,\tau)\right)(-s)^{-\frac{1}{4}}+\mathcal{O}((-s)^{-1}),\nonumber\\
(\beta(\Vec{s}))_{12}&=\overline{(\beta(\Vec{s}))_{21}}=2^{-\frac{3}{4}}\mathrm{e}^{2\sqrt{2}\tau i}i(-s)^{-\frac{1}{4}}\bigg(\overline{a_0(4\sqrt{2}\tau i)b_1(4\sqrt{2}\tau i)}(k_1(s)+\overline{k_2}(s))\\
&+a_0(4\sqrt{2}\tau i)b_1(4\sqrt{2}\tau i)y^2(4\sqrt{2}\tau i)u(4\sqrt{2}\tau i)\frac{v(4\sqrt{2}\tau i)-2\alpha}{v(4\sqrt{2}\tau i)}(\overline{k_1}(s)+k_2(s))\bigg)+\mathcal{O}((-s)^{-1}).
 \end{align}
    Moreover, ${H}_{II}(\Vec{s})$ defined in \eqref{eq:def delta(s)} has the following asymptotic behavior as $s\to-\infty$, 
  
       \begin{align}
        ({H}_{II}(\Vec{s}))_{jj}&=2\sqrt{2}\sqrt{-s}\left((-1)^{j+1}H_V(4\sqrt{2}\tau i)-\frac{\rho_1+\rho_2}{2i}\right)+\mathcal{O}((-s)^{-1}), ~~j=1,2,\\
   ({H}_{II}(\Vec{s}))_{12}&=2\sqrt{2}\sqrt{-s}\mathrm{e}^{2\sqrt{2}\tau i}\frac{y(4\sqrt{2}\tau i)}{4\sqrt{2}\tau i}(v(4\sqrt{2}\tau i)-2\alpha)(1-u(4\sqrt{2}\tau i))+\mathcal{O}((-s)^{-1}),  \\     
           ({H}_{II}(\Vec{s}))_{21}&=2\sqrt{2}\sqrt{-s}\mathrm{e}^{-2\sqrt{2}\tau i}\frac{1}{y(4\sqrt{2}\tau i)4\sqrt{2}\tau i}v(4\sqrt{2}\tau i)\left(1-\frac{1}{u(4\sqrt{2}\tau i)}\right)+\mathcal{O}((-s)^{-1}), \label{eq:asy Ham for NC PII}
    \end{align}
  where 
  \begin{equation}
     \begin{split}
        \varphi_j(s,\tau)&=\frac{4}{3}\sqrt{2}(-s)^{\frac{3}{2}}-\frac{\pi}{4}-\arg\Gamma(-i{\rho_j})-{\rho_j}\ln(2^{\frac{9}{2}}(-s)^{\frac{3}{2}})\\
        &+(-1)^j\arg \bigg(a_0(4\sqrt{2}\tau i)b_1(4\sqrt{2}\tau i)y(4\sqrt{2}\tau i)\frac{v(4\sqrt{2}\tau i)-2\alpha}{v(4\sqrt{2}\tau i)}\bigg),
        \end{split}
     \end{equation}
     \begin{equation}
  \begin{split}
  \psi_j(s,\tau)&=\frac{4}{3}\sqrt{2}(-s)^{\frac{3}{2}}-\frac{\pi}{4}-\arg\Gamma\left(-i{\rho_j}\right)-{\rho_j}\ln(2^{\frac{9}{2}}(-s)^{\frac{3}{2}})\\
  &+(-1)^j\arg \left(a_0(4\sqrt{2}\tau i)b_1(4\sqrt{2}\tau i)y(4\sqrt{2}\tau i)u(4\sqrt{2}\tau i)\right),
  \end{split}
   \end{equation}
\begin{equation}
 k_j(s)=\sqrt{\rho_j}\mathrm{exp}\left({\frac{4}{3}\sqrt{2}i(-s)^{\frac{3}{2}}-i\arg\Gamma\left(-i{\rho_j}\right)-\frac{\pi}{4}i-i{\rho_j}\ln(2^{\frac{9}{2}}(-s)^{\frac{3}{2}})}\right),~~j=1,2.
\end{equation}
Here $u(x),~v(x)$ and $y(x)$  are solutions of the PV system \eqref{eq:PV for u}-\eqref{eq:PV for y} with the properties specified in Theorem \ref{thm:2} and $a_0(x),~b_1(x)$ are defined by the relation \eqref{eq:analyticfac}.
\end{thm}

\begin{remark}
    We mention that the functions $v$, $y$ and $u$
in the denominators in the asymptotic formulas in  Theorem \ref{thm5} may not be nonzero on
$i\mathbb{R}$. Nevertheless, each term involving  denominator is analytic for $x\in i\mathbb{R}$. For instance, we have  $$a_0(x)b_1(x)y(x)^2u(x)\frac{v(x)-2\alpha}{v(x)}=(\Phi_0^{(0)}(x))_{11}(\Phi_0^{(1)}(x))_{12},$$ where $\Phi_0^{(0)}(x),\Phi_0^{(1)}(x)$ are defined in \eqref{eq:analyticfac}. By Proposition \ref{pro:pro1}, these function are free of poles for $x\in i\mathbb{R}$.
Similarly, the remaining terms can be shown to be free of poles for $x\in i\mathbb{R}$ by using \eqref{eq:analyticfac} and Proposition \ref{pro:pro1}. \end{remark}

The rest of this paper is planned as follows. In Section \ref{RHforFD and PV}, we introduce the Riemann-Hilbert (RH for short) problem $\Psi(\lambda,\Vec{s})$ for the noncommutative PII equation, which is related to the logarithmic derivative of the Fredholm determinant of the matrix Airy operator \eqref{eq:Mat KAiry}, as shown in \cite{ber12}. We also formulate the RH problem $\Phi(\lambda,x)$ for the PV equation, which will be used in constructing the global parametrix in the asymptotic analysis of the RH problem for  $\Psi(\lambda,\Vec{s})$. We then establish the solvability of $\Phi(\lambda,x)$ for the parameter $\alpha\in i\mathbb{R}$ and $x\in i\mathbb{R}$ by proving a vanishing lemma. The solvability implies the pole-free of the PV transcendents for $x\in i\mathbb{R}$. In Sections \ref{sec:asy of PV for large x} and \ref{sec:asy of PV for small x}, we derive the asymptotic behavior of the PV transcendents as $ix\to -\infty$ and $x\to0$, respectively, by using the Deift-Zhou nonlinear steepest descent method \cite{DIZ, DMVZ1, DMVZ2, DZ} for the RH problem for $\Phi(\lambda,x)$. 
In Section \ref{sec: large t}, we derive the asymptotic approximations for the logarithmic derivative of the Fredholm determinant and a particular solution of the noncommutative PII equation as $s\to -\infty$ also by performing the Deift-Zhou method for the RH problem for $\Psi(\lambda,\Vec{s})$. The proofs of Theorems \ref{thm:4} and \ref{thm5} will be presented in Section \ref{sec:proof 4 and 5}.

\section{Model Riemann-Hilbert problem}\label{RHforFD and PV}
In this section, we first relate the logarithmic derivative of the Fredholm determinant \eqref{eq:int of FD} and the noncommutative PII equation \eqref{eq:NC PII} to the solution of a RH problem for $\Psi(\lambda,\Vec{s})$, as shown in \cite{ber12}. Then, we present the special PV equations and their related model RH problem encountered in our study.

\subsection{Model RH problem for $\Psi$}
In the following, given any $n\times n$ matrix $A$ 
 and $2\times 2$ matrix $G=(g_{ij})_{i,j=1}^2$, we indicate with the tensor product $A\otimes G$ the matrix
\begin{equation}\label{eq:tensor}
    A\otimes G:=\begin{pmatrix}
        g_{11}A & g_{12}A\\
        g_{21}A & g_{22}A
    \end{pmatrix}.
\end{equation}
The symbols $\hat{\sigma}_i,i=1,2,3,$ will denote the matrices $I\otimes\sigma_i$, where $\sigma_i,i=1,2,3$ are the Pauli matrices
$$
\sigma_1=\begin{pmatrix}0 & 1\\
1 & 0    
\end{pmatrix},~
\sigma_2=\begin{pmatrix}
    0 & i\\
    -i & 0
\end{pmatrix},~
\sigma_3=\begin{pmatrix}
    1 & 0\\
    0 & -1
\end{pmatrix}.
$$
\subsubsection*{RH problem for $\Psi$}
\begin{description}
    \item(1) $\Psi(\lambda)=\Psi(\lambda,\Vec{s})$ is a $4\times 4$ matrix-valued function, which is analytic in $\mathbb{C}\backslash\Sigma_k,~ k=1,2,3,4$; see Figure \ref{fig:fig1} for the contours.
    \begin{figure}
    \centering
    \includegraphics[width = 0.6\textwidth]{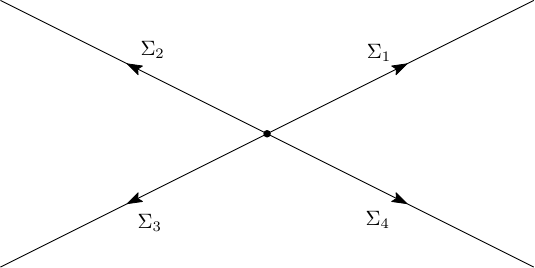}
    \caption{The contours of the RH problem for $\Psi$.}
   \label{fig:fig1}
\end{figure}
\item(2) $\Psi(\lambda)$ satisfies the jump condition
\begin{equation}
\Psi_+(\lambda)= \Psi_-(\lambda)S_k,~ \lambda\in \Sigma_k,~ k=1,2,3,4,
\end{equation}
where $\Psi_+$ and $\Psi_-$ denote the limits of the function $\Psi$ on the ray $\Sigma_k$ from the left and right hand side, respectively. Here the jump matrices are of the form 
\begin{equation}
S_1=\begin{pmatrix}I & 0\\
C & I\end{pmatrix},~
S_2=\begin{pmatrix}I & 0\\
-C & I\end{pmatrix},~
S_3=\begin{pmatrix}I & C\\
0 & I\end{pmatrix},~
S_4=\begin{pmatrix}I & -C\\
0 & I\end{pmatrix},
\end{equation}
where $I$ is a $2 \times 2$ identity matrix, and $C$ is the $2 \times 2$ Hermitian matrix defined in \eqref{eq:def C}.
\item(3) As $\lambda\to\infty$, we have
\begin{equation}
\Psi(\lambda) = \left(I_4+\mathcal{O}(\lambda^{-1})\right)\mathrm{e}^{-\Theta(\lambda)},  
\end{equation}
where $I_4$ is a $4 \times 4$ identity matrix and
\begin{equation}
\Theta(\lambda):=\theta(\lambda)\otimes\sigma_3=\begin{pmatrix}
    \theta(\lambda) & 0\\
    0 & -\theta(\lambda)
\end{pmatrix},~
\theta(\lambda)=\begin{pmatrix}\frac{i}{6}\lambda^3+is_1\lambda & 0\\
0 & \frac{i}{6}\lambda^3+is_2\lambda
\end{pmatrix}.
\end{equation}
\end{description}

Then, the logarithmic derivative of the Fredholm determinant \eqref{eq:int of FD} and the associated solutions to the noncommutative PII equation \eqref{eq:NC PII}  can be expressed in terms of the solution of the above RH problem  \cite{ber12}:  
\begin{equation}\label{eq:RH for NC PII}
    \beta(\Vec{s})=-i \lim \limits_{\lambda \to \infty} \lambda(\Psi(\lambda)\mathrm{e}^{\Theta(\lambda)})_{21}, ~~{H}_{II}(\Vec{s})=\lim_{\lambda\to\infty}\lambda(\Psi(\lambda)\mathrm{e}^{\Theta(\lambda)}-I_4)_{22},
\end{equation}
\begin{equation}\label{eq:RH for FD}
    D\ln\det(Id-Ai_{\Vec{s}}^2)=-2iTr\left({H}_{II}(\Vec{s})\right),~D=\frac{\partial}{\partial_{s_1}}+\frac{\partial}{\partial_{s_2}},
\end{equation}
where $\Psi=\begin{pmatrix}
    \Psi_{11} & \Psi_{12}\\
    \Psi_{21} & \Psi_{22}
\end{pmatrix}$, $\Psi_{ij}$ denote the $2 \times 2$ blocks
of $\Psi,~ i,j=1,2$, $\beta(\Vec{s})$ and
${H}_{II}(\Vec{s})$ are defined in \eqref{eq:NC PII} and \eqref{eq:def delta(s)}, respectively.

\subsection{Model RH problem for the PV equation}
For our application, we introduce the following RH problem with specific parameter for the PV equation $\Phi(\lambda,x;\alpha)$ \cite{FIKY,XZ}, which solves the ODE system \eqref{eq:lax pair L} and \eqref{eq:lax pair U}.
\subsubsection*{RH problem for $\Phi$}
\begin{description}
\item(1)
$\Phi(\lambda)=\Phi(\lambda,x;\alpha)$ is a $2\times 2$ matrix-valued function, which is analytic for $\lambda\in\mathbb{C}\backslash[0,1]$.
\item(2)
$\Phi(\lambda)$ satisfies the jump condition
\begin{equation}\label{eq:jump PV}
\Phi_+(\lambda)=\Phi_-(\lambda)U\mathrm{e}^{2\pi i\alpha\sigma_3}U^{-1}, ~\lambda\in(0,1), 
\end{equation}
where $\alpha\in i\mathbb{R}$ and $U$ is a unitary matrix defined in \eqref{eq:def U}. 
\item(3)
The asymptotic behavior of $\Phi(\lambda)$ at infinity is
\begin{equation}\label{eq:asyofPVinfty}
\Phi(\lambda)=\left(I+\frac{\Phi_{-1}(x)}{\lambda}+\mathcal{O}(\lambda^{-2})\right)\mathrm{e}^{-\frac{1}{2}\lambda x\sigma_3}, \lambda\to\infty,
\end{equation}
where
\begin{equation}\label{eq:RH Ham for PV}
    \Phi_{-1}(x)=\begin{pmatrix}
        -H_V & -\frac{y}{x}(v-2\alpha)(1-u)\\
        -\frac{1}{yx}v(1-\frac{1}{u}) & H_V
    \end{pmatrix},
\end{equation}
and $H_V$ is the Hamiltonian defined in \eqref{eq:def for HV}.
\item(4)
The behavior of $\Phi(\lambda)$ at $\lambda=0$ is
\begin{equation}\label{eq:local0}
\Phi(\lambda)={\Phi}^{(0)}(\lambda,x)\lambda^{-\alpha\sigma_3}\mathrm{e}^{\pi i\alpha\sigma_3}U^{-1}, 
\end{equation}
as $\lambda\to 0$, where $U$ is a unitary matrix independent of $\lambda$ defined in \eqref{eq:def U}, and we choose the branch of $\lambda^{\alpha}$ such that $\arg\lambda\in(0,2\pi)$. Here the function ${\Phi}^{(0)}(\lambda,x)$ is analytic at $\lambda=0$ and satisfies the following expansion
\begin{equation}\label{eq:expansion Phi 0}
    {\Phi}^{(0)}(\lambda,x)={\Phi}^{(0)}_0(x)\left(I+{\Phi}^{(0)}_1(x) \lambda+\mathcal{O}(\lambda^2)\right).
\end{equation}
\item (5)
The behavior of $\Phi(\lambda)$ at $\lambda=1$ is
\begin{equation}\label{eq:local1}
\Phi(\lambda)={\Phi}^{(1)}(\lambda,x)(\lambda-1)^{\alpha\sigma_3}U^{-1}
\end{equation}
as $\lambda\to 1$, where $(\lambda-1)^{\alpha}$ takes the principal branch. The function ${\Phi}^{(1)}(\lambda,x)$ is analytic at $\lambda=1$ and satisfies the following expansion
\begin{equation}\label{eq:expansion Phi 1}
    {\Phi}^{(1)}(\lambda,x)={\Phi}^{(1)}_0(x)\left(I+{\Phi}^{(1)}_1(x) (\lambda-1)+\mathcal{O}((\lambda-1)^2)\right).
\end{equation}
\end{description}
The solutions to the PV equations can be obtained from 
${\Phi}^{(0)}_0(x)$ and ${\Phi}^{(1)}_0(x)$; cf. \cite{FIKY}, 
\begin{equation}\label{eq:analyticfac}
{\Phi}^{(0)}_0(x)=\begin{pmatrix}
    a_0y\frac{v-2\alpha}{v} & b_0y\\
    a_0 & b_0
\end{pmatrix},~
{\Phi}^{(1)}_0(x)=\begin{pmatrix}
    a_1 & b_1yu\\
    \frac{a_1}{yu}\frac{v}{v-2\alpha} & b_1
\end{pmatrix},
\end{equation}
with
\begin{equation}\label{eq:det}
a_0b_0y=\frac{v}{-2\alpha},~~a_1b_1=\frac{v-2\alpha}{-2\alpha},~~\alpha\neq 0.
\end{equation}
Besides, the PV transcendents can also be recovered from $\Phi_{-1}(x)$ in \eqref{eq:RH Ham for PV} via 
\begin{equation}\label{eq:RH for v1}
    x^2(\Phi_{-1})_{12}(\Phi_{-1})_{21}=(1-u)(v-2\alpha)(v-\frac{v}{u}),
\end{equation}
\begin{equation}\label{eq:RH for v2}
    v(x)=(\Phi_{-1})_{11}+x(\Phi_{-1})_{12}(\Phi_{-1})_{21}+\alpha.
\end{equation}
\begin{remark}\label{re:C diagonal}
 If the matrix $U$ in \eqref{eq:jump PV} is a diagonal matrix, the solution for $\Phi(\lambda)$ can be obtained explicitly:
 \begin{equation*}
\Phi(\lambda)=\left(\frac{\lambda}{\lambda-1}\right)^{\alpha\sigma_3}\mathrm{e}^{-\frac{1}{2}\lambda x\sigma_3}.
 \end{equation*} Thus, we have ${\Phi}^{(0)}_0(x)=U$ and ${\Phi}^{(1)}_0(x)=\mathrm{e}^{-\frac{1}{2}x\sigma_3}U$. From the definitions of $H_V(x)$ \eqref{eq:def for HV} and $\hat{u}(x)$ \eqref{eq:def hatu}, we obtain 
 \begin{equation*}
     H_V(x)=\alpha,~v(x)=0,~\hat{u}(x)=\mathrm{e}^{\frac{x}{2}}.
 \end{equation*}
 If the parameter $\alpha=0$, we then have 
 $$\Phi(\lambda)=\mathrm{e}^{-\frac{1}{2}\lambda x\sigma_3},$$
 which implies
 $$H_V(x)=0, ~v(x)=0, ~\hat{u}(x)=|\omega|^2\mathrm{e}^{\frac{x}{2}}+(1-|\omega|^2){e}^{-\frac{x}{2}}.$$
\end{remark}
\begin{remark}\label{re:symmetry -x}
    Since $\mathrm{e}^{-\frac{1}{2}x\sigma_3}\Phi(1-\lambda,-x;-\alpha)$ satisfies the same RH problem as $\Phi(\lambda,x;\alpha)$, it then follows from the unique solvability of the RH problem for $\Phi$, given in Proposition \ref{pro:pro1}, that
\begin{equation}
    \Phi(\lambda,x;\alpha)=\mathrm{e}^{-\frac{1}{2}x\sigma_3}\Phi(1-\lambda,-x;-\alpha).
\end{equation}
Then we have
\begin{equation}\label{eq:symmetry x and -x}
    \mathrm{e}^{-\frac{1}{2}x\sigma_3}\Phi^{(1)}_0(-x;-\alpha)=\Phi^{(0)}_0(x;\alpha),~    \mathrm{e}^{-\frac{1}{2}x\sigma_3}\Phi^{(0)}_0(-x;-\alpha)=\Phi^{(1)}_0(x;\alpha).
\end{equation}
Combining this with \eqref{eq:analyticfac}, we obtain the following symmetry relation for this family of solutions to the PV equation 
\begin{equation}
    u(-x)=u(x)^{-1},~v(-x)=-v(x),~y(-x)=\mathrm{e}^xy(x)u(x),~x\in i\mathbb{R}.
\end{equation}
\end{remark}
For later use, we calculate the trace of the matrices ${\Phi}^{(1)}_1(x)$ and ${\Phi}^{(0)}_1(x)$. By substituting \eqref{eq:local0} into \eqref{eq:lax pair L}, we find
\begin{align}
        &-\alpha{\Phi}^{(0)}_0(x)\sigma_3=L_1(x){\Phi}^{(0)}_0(x),\label{eq:Phi_10}\\
    &\left(-\frac{x}{2}\sigma_3-L_2(x)\right){\Phi}^{(0)}_0(x)+L_1(x){\Phi}^{(0)}_0(x){\Phi}^{(0)}_1(x)={\Phi}^{(0)}_0(x){\Phi}^{(0)}_1(x)(I-\alpha\sigma_3).\label{eq:Phi_101}
\end{align}
Then, we obtain 
\begin{equation}\label{eq:trace0}
Tr\left({\Phi}^{(0)}_1(x)\right)=Tr\left({\Phi}^{(0)}_0(x)^{-1}\left(-\frac{x}{2}\sigma_3-L_2(x)\right){\Phi}^{(0)}_0(x)\right),
\end{equation}
and
\begin{equation}\label{eq:trace1}
    Tr\left({\Phi}^{(1)}_1(x)\right)=Tr\left({\Phi}^{(1)}_0(x)^{-1}\left(-\frac{x}{2}\sigma_3+L_1(x)\right){\Phi}^{(1)}_0(x)\right).
\end{equation}

\subsection{Vanishing lemma and existence of solution to the model RH problem for $\Phi$}
We will prove the existence of solution to the model RH problem for $\Phi(\lambda,x;\alpha)$ if the parameter $\alpha\in i\mathbb{R}$ and $x\in i\mathbb{R}$. We start with the vanishing lemma which shows that the homogeneous RH problem has only zero solution.
\begin{lem}
    For  $\alpha\in i\mathbb{R}$ and $x\in i\mathbb{R}$, we suppose that $\hat{\Phi}(\lambda)$ satisfies the same jump condition \eqref{eq:jump PV} and the same behaviors \eqref{eq:local0} and \eqref{eq:local1} as $\Phi(\lambda)$, respectively, as $\lambda$ tends to the origin and $1$. Further assume the behavior of $\hat{\Phi}(\lambda)$ to be 
    \begin{equation}
        \hat{\Phi}(\lambda)=\mathcal{O}\left(\frac{1}{\lambda}\right)\mathrm{e}^{-\frac{1}{2}\lambda x\sigma_3},~~\lambda\to\infty.
    \end{equation}
    Then, we have $\hat{\Phi}(\lambda)=0$ for $\lambda\in\mathbb{C}.$
\end{lem}
\begin{proof}
    To normalize the behavior as $\lambda\to\infty$, we introduce the transformation
    \begin{equation}
        \Tilde{\Phi}(\lambda)=\hat{\Phi}(\lambda)\mathrm{e}^{\frac{1}{2}\lambda x\sigma_3}.
    \end{equation}
    Then we have the normalized behavior at infinity
    \begin{equation}
        \Tilde{\Phi}(\lambda)=\mathcal{O}\left(\frac{1}{\lambda}\right),
    \end{equation}
    and the modified jump condition 
    \begin{equation}\label{eq:jump TPhi}
        \Tilde{\Phi}_+(\lambda)= \Tilde{\Phi}_-(\lambda)
        \Tilde{J}(\lambda),~\lambda \in (0,1),
    \end{equation}
    where
    \begin{equation}
        \Tilde{J}(\lambda)=\mathrm{e}^{-\frac{1}{2}\lambda x\sigma_3}U\mathrm{e}^{2\pi i\alpha\sigma_3}U^{-1}\mathrm{e}^{\frac{1}{2}\lambda x\sigma_3}.
    \end{equation}
    Moreover, $\Tilde{\Phi}(\lambda)$ is analytic in $\mathbb{C}\backslash[0,1]$ and fulfills the asymptotic behaviors
    \begin{equation}
        \Tilde{\Phi}(\lambda)=\mathcal{O}(1)\lambda^{-\alpha\sigma_3},~\lambda\to 0,
    \end{equation}
    and 
    \begin{equation}
         \Tilde{\Phi}(\lambda)=\mathcal{O}(1)(\lambda-1)^{\alpha\sigma_3},~\lambda\to 1.
    \end{equation}
    Next, we consider 
    \begin{equation}
        Q(\lambda)= \Tilde{\Phi}(\lambda) (\Tilde{\Phi}(\overline{\lambda}))^*,~\lambda\in\mathbb{C}\backslash \mathbb{R},
    \end{equation}
    where $(\Phi(\lambda))^*$ denotes the Hermitian conjugate of $ \Phi(\lambda)$. We see that $Q(\lambda)$ is analytic in $\mathbb{C}\backslash \mathbb{R}$, and we have
    \begin{equation}\label{eq:int Q}
        \int _{\mathbb{R}} Q_+(x)dx=\int _\mathbb{R}\Tilde{\Phi}_+(x)\hat{J}(x)(\Tilde{\Phi}_+(x))^* dx=0,
    \end{equation}
    where
    \begin{equation}
        \hat{J}(x)=\begin{cases}
            (\Tilde{J}^{-1}(x))^{*},~&x\in(0,1),\\
            I,~&x\in\mathbb{R}\backslash[0,1].
        \end{cases}
    \end{equation}
    For $\alpha\in i\mathbb{R}$ and $x\in i\mathbb{R}$, the jump matrix $\hat{J}(\lambda)$ is a Hermitian positive definite matrix. Thus, if we define $\boldsymbol{a}(x)=(\Tilde{\Phi}_{11,+}(x),\Tilde{\Phi}_{12,+}(x))$, we have $\boldsymbol{a}(x)\hat{J}(x)\boldsymbol{a}(x)^{*}\ge 0,~x\in\mathbb{R}$.
    This, along with \eqref{eq:int Q}, implies that the first row of $\Tilde{\Phi}(\lambda)$ vanishes in the upper-half of the $\lambda$-plane. Combining this with the jump condition \eqref{eq:jump TPhi}, we conclude that the first row of $\Tilde{\Phi}(\lambda)$ vanishes in the whole plane. Similarly, the second row of $\Tilde{\Phi}(\lambda)$ also vanishes in the whole plane. This completes the proof of the lemma.
\end{proof}

By a standard analysis \cite{DMVZ1, FIKY, FMZ, IIKS}, the vanishing lemma implies the existence of unique solution to the RH problem for $\Phi(\lambda,x)$ for $\alpha\in i\mathbb{R}$ and $x\in i\mathbb{R}$.
\begin{pro}\label{pro:pro1}
    For  $\alpha\in i\mathbb{R}$ and $x\in i\mathbb{R}$, there exists unique solution to the RH problem for  $\Phi(\lambda,x)$. Then $\Phi(\lambda,x)$ is free of poles for $x\in i\mathbb{R}$. Particularly, the solutions to the PV equations $u(x)$, $v(x)$, $y(x)$ defined in \eqref{eq:PV for u}-\eqref{eq:PV for y} and the associated Hamiltonian $H_V(x)$ \eqref{eq:def for HV} are free of poles for $x\in i\mathbb{R}$.
\end{pro}

At the end of this section, we prove the following special properties of the PV trancendents.
\begin{pro}\label{pro:2}
    Let $\alpha\in i\mathbb{R}$, then the  PV transcendents $u(x)$, $v(x)$ and the Hamiltonian $H_V(x)$ defined in \eqref{eq:PV for u}, \eqref{eq:PV for v} and \eqref{eq:def for HV} under the conditions \eqref{eq:para condition} satisfy the following properties:
    \begin{equation}
        |u(x)|=1,~~\Re(v(x))=0,~~\Re(H_V(x))=0,~~x\in i\mathbb{R}.
    \end{equation}
\end{pro}
\begin{proof}
From the unique solvability of the RH problem for $\Phi$, we have 
    \begin{equation}
        \Phi^{-1}(\lambda)=(\Phi(\overline{\lambda}))^*.
    \end{equation}
 Then, from \eqref{eq:analyticfac}, we can obtain 
    \begin{align}
        \begin{pmatrix}\label{eq:symmetry 0}
            a_0y\frac{v-2\alpha}{v} & b_0y\\
    a_0 & b_0
        \end{pmatrix}=\begin{pmatrix}
            \overline{b_0} & -\overline{a_0}\\
            -\overline{b_0y} & \overline{a_0y\frac{v-2\alpha}{v} }
        \end{pmatrix},\\
        \begin{pmatrix}
            a_1 & b_1yu\\
    \frac{a_1}{yu}\frac{v}{v-2\alpha} & b_1
        \end{pmatrix}=\begin{pmatrix}\label{eq:symmetry 1}
            \overline{b_1} & -\overline{\frac{a_1}{yu}\frac{v}{v-2\alpha}}\\
            -\overline{b_1yu} & \overline{a_1}
        \end{pmatrix}.
    \end{align}
    So we have 
    \begin{align}
        \frac{|y(x)|^2}{v(x)}(v(x)-2\alpha)=-1,\label{eq:identity1}\\
        \frac{v(x)}{v(x)-2\alpha}=-|y(x)|^2|u(x)|^2.\label{eq:identity2}
    \end{align}
    Then the condition $\Re(v(x))=0$ can be deduced from \eqref{eq:identity1}, since $\alpha$ is a pure imaginary number. Combining \eqref{eq:identity1} and \eqref{eq:identity2}, we see that $|u(x)|^2=1$. This, together with \eqref{eq:RH Ham for PV}, \eqref{eq:RH for v1} and \eqref{eq:RH for v2}, leads to the conclusion that $\Re(H_V(x))=0$ for $x\in i\mathbb{R}$. This completes the proof.
    \end{proof}

\section{Asymptotic analysis of the RH problem for $\Phi$ as $ix\to -\infty$}\label{sec:asy of PV for large x}
In this section, we derive the asymptotic expansions of the solutions to the PV equation as $ix\to-\infty$ by performing the Deift-Zhou nonlinear steepest descent analysis \cite{DIZ, DMVZ1, DMVZ2, DZ} of the RH problem for $\Phi(\lambda,x)$.
\subsection{Normalization}
To normalize the behavior of $\Phi(\lambda,x)$ at infinity, we will introduce the following first transformation 
\begin{equation}
    Y(\lambda)=\Phi(\lambda)\mathrm{e}^{\frac{1}{2}\lambda x\sigma_3}.
\end{equation}
Then $Y(\lambda)$ satisfies the following RH problem.
\subsubsection*{RH problem for $Y$}
\begin{description}
    \item(1)
    $Y(\lambda)$ is analytic for $\lambda\in \mathbb{C}\backslash[0,1]$.
\item (2)
$Y(\lambda)$ satisfies the jump condition
\begin{equation}
    Y_+(\lambda)=Y_-(\lambda)\mathrm{e}^{-\frac{1}{2}\lambda x\sigma_3}U\mathrm{e}^{2\pi i\alpha\sigma_3}U^{-1}\mathrm{e}^{\frac{1}{2}\lambda x\sigma_3},~~\lambda\in(0,1).
\end{equation}
\item (3)
The asymptotic behavior of $Y(\lambda)$ at infinity is
\begin{equation}
    Y(\lambda)=I+\mathcal{O}(\lambda^{-1}), ~~\lambda\to \infty.
\end{equation}
\item (4)
The behavior of $Y(\lambda)$ at $\lambda=0$ is
\begin{equation}\label{eq:Ylocal0}
Y(\lambda)={\Phi}^{(0)}(\lambda,x)\lambda^{-\alpha\sigma_3}\mathrm{e}^{\pi i\alpha\sigma_3}U^{-1}\mathrm{e}^{\frac{1}{2}\lambda x\sigma_3},~~\lambda\to 0.
\end{equation}
\item (5)
The behavior of $Y(\lambda)$ at $\lambda=1$ is
\begin{equation}\label{eq:Ylocal1}
Y(\lambda)={\Phi}^{(1)}(\lambda,x)(\lambda-1)^{\alpha\sigma_3}U^{-1}\mathrm{e}^{\frac{1}{2}\lambda x\sigma_3},~~\lambda\to 1.
\end{equation}
\end{description}
\subsection{Opening the lenses and deformation}
Denote $M=U\mathrm{e}^{2\pi i \alpha\sigma_3}U^{-1}=(m_{ij})_{i,j=1}^2$, 
it follows from the fact that $\alpha\in i\mathbb{R}$ and $U$ is a unitary matrix that  $M$ is a Hermitian matrix and  admits a unique Gauss UL-decomposition:
\begin{equation}
    M=S_US_DS_L=\begin{pmatrix}
        1 & \frac{m_{12}}{m_{22}}\\
        0 & 1
    \end{pmatrix}
    \begin{pmatrix}
        \frac{1}{m_{22}} & 0\\
        0 & m_{22}
    \end{pmatrix}
    \begin{pmatrix}
        1 & 0\\
        \frac{\overline{m_{12}}}{m_{22}} & 1
    \end{pmatrix},
\end{equation}
where 
\begin{equation}\label{eq:def M}
m_{22}=(1-|\omega|^2)\mathrm{e}^{2\pi i\alpha}+|\omega|^2\mathrm{e}^{-2\pi i\alpha},~m_{12}=\omega\sqrt{1-|\omega|^2}(\mathrm{e}^{-2\pi i\alpha}-\mathrm{e}^{2\pi i\alpha}).  
\end{equation}
Based on the factorization of jumps and deformation of contours, we introduce the second transformation
\begin{equation} 
T(\lambda)=Y(\lambda)\begin{cases}
\mathrm{e}^{-\frac{1}{2}\lambda x\sigma_3}S_L^{-1}\mathrm{e}^{\frac{1}{2}\lambda x\sigma_3},& \lambda\in\Omega_1,\\
\mathrm{e}^{-\frac{1}{2}\lambda x\sigma_3}S_U\mathrm{e}^{\frac{1}{2}\lambda x\sigma_3},& \lambda\in\Omega_2,\\
I& elsewhere;
\end{cases}
\end{equation}
see Figure \ref{fig:fig2} for an illustration.
Then $T(\lambda)$ satisfies the following RH problem.
\subsubsection*{RH problem for $T$}
\begin{description}
    \item(1)
    $T(\lambda)$ is analytic for $\lambda\in \mathbb{C}\backslash\Sigma_T$, where the contours are shown in Figure \ref{fig:fig2}.
    \begin{figure}[h]
    \centering
    \includegraphics[width=0.5\textwidth]{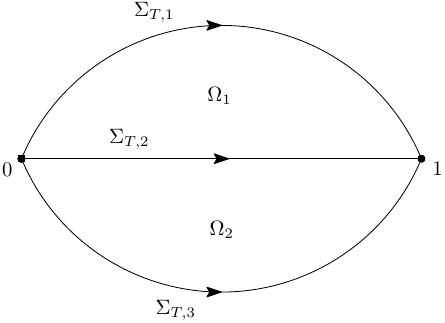}
    \caption{The contours of the RH problem for $T$}
    \label{fig:fig2}
\end{figure}
\item (2)
$T(\lambda)$ satisfies the jump condition
\begin{equation}\label{eq:jump JT}
    T_+(\lambda)=T_-(\lambda)J_T(\lambda),~ \lambda\in\Sigma_T,
\end{equation}
where
\begin{equation}\label{eq:defJT}
    J_T(\lambda)=\begin{cases}
        \mathrm{e}^{-\frac{1}{2}\lambda x\sigma_3}S_L\mathrm{e}^{\frac{1}{2}\lambda x\sigma_3}, & \lambda\in\Sigma_{T,1},\\
        S_D, & \lambda\in\Sigma_{T,2},\\
        \mathrm{e}^{-\frac{1}{2}\lambda x\sigma_3}S_U\mathrm{e}^{\frac{1}{2}\lambda x\sigma_3},& \lambda\in\Sigma_{T,3.}
    \end{cases}
\end{equation}
\item (3)
As $\lambda\to\infty$, we have
\begin{equation}
    T(\lambda)=I+\mathcal{O}\left(\lambda^{-1}\right).
\end{equation}
\item (4)
The behavior of $T(\lambda)$ at $\lambda=0$ is
\begin{equation}\label{eq:Tlocal0}
T(\lambda)={\Phi}^{(0)}(\lambda,x)\lambda^{-\alpha\sigma_3}\mathrm{e}^{\pi i\alpha\sigma_3}U^{-1}\begin{cases}
        S_L^{-1}\mathrm{e}^{\frac{1}{2}\lambda x\sigma_3}, & \lambda\in\Omega_1,\\
        S_U\mathrm{e}^{\frac{1}{2}\lambda x\sigma_3}, & \lambda\in\Omega_2,\\
        \mathrm{e}^{\frac{1}{2}\lambda x\sigma_3} & eleswhere
    \end{cases}
,~~\lambda\to 0.
\end{equation}
\item (5)
The behavior of $T(\lambda)$ at $\lambda=1$ is
\begin{equation}\label{eq:Tlocal1}
T(\lambda)={\Phi}^{(1)}(\lambda,x)(\lambda-1)^{\alpha\sigma_3}U^{-1}\begin{cases}
        S_L^{-1}\mathrm{e}^{\frac{1}{2}\lambda x\sigma_3}, & \lambda\in\Omega_1,\\
        S_U\mathrm{e}^{\frac{1}{2}\lambda x\sigma_3}, & \lambda\in\Omega_2,\\
        \mathrm{e}^{\frac{1}{2}\lambda x\sigma_3} & eleswhere
    \end{cases},~~\lambda\to 1.
\end{equation}
\end{description}

\subsection{Global parametrix}
Since all jump matrices in the RH problem for $T(\lambda)$ tend to the identity matrix as $ix\to -\infty$, except for the one on $(0,1)$, we focus solely on the jump matrix on $(0,1)$. This leads us to the following RH problem for the global parametrix $P^{(\infty)}(\lambda)$.
\subsubsection*{RH problem for $P^{(\infty)}$}
\begin{description}
    \item(1)
    $P^{(\infty)}(\lambda)$ is analytic for $\lambda\in\mathbb{C}\backslash[0,1]$.
\item (2)
$P^{(\infty)}(\lambda)$ satisfies the jump condition 
\begin{equation}
    P_+^{(\infty)}(\lambda)=P_-^{(\infty)}(\lambda)(m_{22})^{-\sigma_3}, ~\lambda\in(0,1).
\end{equation}
\item (3)
As $\lambda\to\infty$, we have 
\begin{equation}
P^{(\infty)}(\lambda)=I+\mathcal{O}(\lambda^{-1}).
\end{equation}
\end{description}

The solution $P^{(\infty)}(\lambda)$ can be constructed explicitly as below
\begin{equation}\label{eq:defPinfty}
    P^{(\infty)}(\lambda)=\left(\frac{\lambda}{\lambda-1}\right)^{\mu\sigma_3},~\mu=\frac{1}{2\pi i}\ln (m_{22}),~\arg(m_{22})\in(-\pi,\pi),
\end{equation}
where the function $\left(\frac{\lambda}{\lambda-1}\right)^{\mu}$ takes the branch cut along $[0,1]$, with its branch fixed by the asymptotic behavior
\begin{equation*}
    \left(\frac{\lambda}{\lambda-1}\right)^\mu\to 1,~\lambda\to \infty.
\end{equation*}

\subsection{Local parametrix at $\lambda=0$}
Near the node point $\lambda=0,1$, the solution $T(\lambda)$ can not be approximated by the global parametrix $P^{(\infty)}(\lambda)$. In this section, we will construct the local parametrix in the neighborhood of $\lambda=0$. Indeed, the local parametrix $P^{(0)}(\lambda)$ at $\lambda=0$ satifies the following RH problem.
\subsubsection*{RH problem for $P^{(0)}$}
\begin{description}
    \item(1)
    $P^{(0)}(\lambda)$ is analytic in $U(0)\backslash\Sigma_T,~U(0)=\lbrace \lambda: |\lambda|<\delta\rbrace$ with $\delta<\frac{1}{2}$.
\item (2)
$P^{(0)}(\lambda)$ has the same jump condition as $T(\lambda)$ for $\lambda\in\Sigma_T\cap U(0).$ 
\item (3)
As $ix\to -\infty$, we have the matching 
condition 
\begin{equation}
    P^{(0)}(\lambda)=\left(I+\mathcal{O}(x^{-1})\right)P^{(\infty)}(\lambda),
\end{equation}
and the error term is uniform for $\lambda\in \partial U(0).$
\item (4)
$P^{(0)}(\lambda)$ remains bounded in a neighborhood of $\lambda=0$.
\end{description}
We are going to solve the RH problem for $P^{(0)}(\lambda)$ by using the confluent hypergeometric parametrix $\Phi^{(CHF)}(\zeta)$ introduced in \ref{sec:Appendix}. For our application, we introduce a matrix-valued parametrix 
$$Z(\zeta)=Z(\zeta;a,b),$$
 which is defined via $\Phi^{(CHF)}(\zeta)$ by the following transformation:
\begin{equation}
    Z(\zeta)=\begin{cases}
        \Phi^{(CHF)}(\zeta)\mathrm{e}^{\frac{1}{2}\pi i(a+b)\sigma_3}, & \zeta\in\Omega_1,\\
        \Phi^{(CHF)}(\zeta)J_3J_4^{-1}\mathrm{e}^{\frac{1}{2}\pi i(b-a)\sigma_3}, & \zeta\in\Omega_2,\\
        \Phi^{(CHF)}(\zeta)J_4^{-1}\mathrm{e}^{\frac{1}{2}\pi i(b-a)\sigma_3}, & \zeta\in\Omega_3,\\
        \Phi^{(CHF)}(\zeta)\mathrm{e}^{\frac{1}{2}\pi i(b-a)\sigma_3}, & \zeta\in\Omega_4,\\
        \Phi^{(CHF)}(\zeta)\begin{pmatrix}
            0 & \mathrm{e}^{\frac{1}{2}\pi ib}\\
            -\mathrm{e}^{-\frac{1}{2}\pi ib} & 0
        \end{pmatrix}\mathrm{e}^{-\frac{1}{2}\pi ia\sigma_3}, & \zeta\in\Omega_5,\\
        \Phi^{(CHF)}(\zeta)J_6\begin{pmatrix}
            0 & \mathrm{e}^{\frac{1}{2}\pi ib}\\
            -\mathrm{e}^{-\frac{1}{2}\pi ib} & 0
        \end{pmatrix}\mathrm{e}^{-\frac{1}{2}\pi ia\sigma_3}, & \zeta\in\Omega_6,\\
        \Phi^{(CHF)}(\zeta)J_7^{-1}J_6\begin{pmatrix}
            0 & \mathrm{e}^{\frac{1}{2}\pi ib}\\
            -\mathrm{e}^{-\frac{1}{2}\pi ib} & 0
        \end{pmatrix}\mathrm{e}^{-\frac{1}{2}\pi ia\sigma_3}, & \zeta\in\Omega_7,\\
        \Phi^{(CHF)}(\zeta)\begin{pmatrix}
            0 & \mathrm{e}^{\frac{1}{2}\pi ib}\\
            -\mathrm{e}^{-\frac{1}{2}\pi ib} & 0
        \end{pmatrix}\mathrm{e}^{\frac{1}{2}\pi ia\sigma_3}, & \zeta\in\Omega_8.
    \end{cases}
\end{equation}
Then $Z(\zeta)$ satisfies the following RH problem.
\subsubsection*{RH problem for $Z$}
\begin{description}
    \item(1)
    $Z(\zeta)$ is analytic for $\zeta\in\mathbb{C}\backslash\lbrace\cup_{k=1}^{3}\Gamma_{Z,k}\rbrace$, where the contours are defined below
\begin{equation*}
    \Gamma_{Z,1}=\mathrm{e}^{\frac{\pi i}{2}}\mathbb{R}^+,~ \Gamma_{Z,2}=\mathrm{e}^{\frac{3\pi i}{4}}\mathbb{R}^+, ~\Gamma_{Z,3}=\mathrm{e}^{\frac{\pi i}{4}}\mathbb{R}^+;
\end{equation*}
see Figure \ref{fig:fig3} for an illustration.
\begin{figure}[h]
    \centering
    \includegraphics[width=0.5\textwidth]{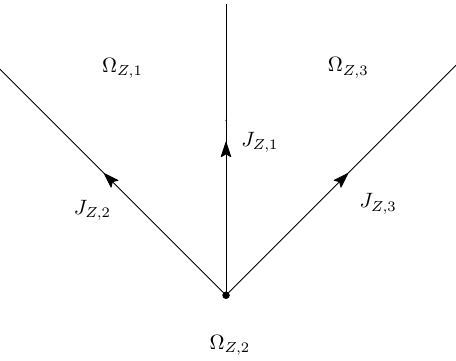}
    \caption{The jump contours and the jump matrices of the RH problem for $Z$.}
    \label{fig:fig3}
\end{figure}
\item (2)
$Z(\zeta)$ satisfies the jump condition
\begin{equation}\label{eq:jumpZ}
    Z_+(\zeta)=Z_-(\zeta)J_{Z,k}, ~\zeta\in\Sigma_{Z,k}, ~k=1,2,3,
\end{equation}
where
\begin{equation}\label{eq:def jump J}
    J_{Z,1}=\mathrm{e}^{2\pi ib\sigma_3},
     ~J_{Z,2}=\begin{pmatrix}
        1 & 0\\
        \mathrm{e}^{\pi i(2b-a)}-\mathrm{e}^{\pi ia} & 1
    \end{pmatrix},~
     J_{Z,3}=\begin{pmatrix}
        1 & \mathrm{e}^{-\pi ia}-\mathrm{e}^{\pi i(2b+a)} \\
        0 & 1
    \end{pmatrix}.
\end{equation}
\item (3)
As $\zeta\to\infty$, we have 
\begin{equation}\label{eq:Z at infity}
   Z(\zeta)=\left(I+\frac{Z_1(a,b)}{\zeta}
   +\mathcal{O}(\zeta^{-2})\right)\zeta^{-b\sigma_3}\mathrm{e}^{-\frac{1}{2}\zeta\sigma_3}\begin{cases}
       \mathrm{e}^{\frac{3}{2}\pi ib\sigma_3}, & \arg(\zeta)\in(-\frac{3\pi}{2},-\frac{\pi}{2}),\\
       
        \mathrm{e}^{-\frac{1}{2}\pi ib\sigma_3}, & \arg(\zeta)\in(-\frac{\pi}{2},\frac{\pi}{2}),
   \end{cases}
\end{equation}
where\begin{equation}\label{eq:Z1}
     Z_1(a,b)=(a^2-b^2)\begin{pmatrix}
       1 & \frac{\Gamma(a-b)}{\Gamma(a+b+1)}\\
       -\frac{\Gamma(a+b)}{\Gamma(a-b+1)} & -1
   \end{pmatrix}   
\end{equation} 
and the branch cut for $\zeta^{-b}$ is taken along the positive imaginary axis such that $\arg \zeta\in\left(-\frac{3\pi}{2},\frac{\pi}{2}\right)$.
\item (4)
As $\zeta\to 0, ~Z(\zeta)$ has the following asymptotic behavior:
\begin{equation}\label{eq:localZ0}
    Z(\zeta)=Z^{(0)}(\zeta)\zeta^{a\sigma_3}B_1B_2^{-1}\begin{cases}
        \mathrm{e}^{\frac{1}{2}\pi ib\sigma_3}, & \zeta\in\Omega_{Z,1},\\
       \mathrm{e}^{\frac{1}{2}\pi ib\sigma_3}J_{Z,2},& \zeta\in\Omega_{Z,2},\\
        \mathrm{e}^{-\frac{3}{2}\pi ib\sigma_3}, & \zeta\in\Omega_{Z,3},
    \end{cases}
\end{equation}
where $Z^{(0)}(\zeta)$ is an analytic factor near the origin,
$\Omega_{Z,k},~k=1,2,3$ are shown in Figure \ref{fig:fig3}, and the branch cut for $\zeta^a$ is chosen such that $\arg\zeta\in\left(-\frac{3\pi}{2},\frac{\pi}{2}\right)$. The constant matrix 
\begin{equation}\label{eq:defC1}
    B_1=\begin{pmatrix}
        1 & \frac{\sin(\pi(a+b))}{\sin(2\pi a)}\\
        0 & 1
    \end{pmatrix},~
    B_2=\begin{pmatrix}
        1 & 0\\
        \mathrm{e}^{-\pi i(a-b)} & 1
    \end{pmatrix}.
\end{equation}
Here the function 
\begin{small}
\begin{equation}\label{eq:defZ0}
\begin{split}
    &Z^{(0)}(\zeta)\\
    &=\mathrm{e}^{\frac{1}{2}b\pi i\sigma_3}\mathrm{e}^{-\frac{\zeta}{2}}\begin{pmatrix}
        \mathrm{e}^{-\frac{\pi i}{2}(a+b)}\frac{\Gamma(1+a-b)}{\Gamma(1+2a)}\phi(a+b,1+2a,\zeta) & -\mathrm{e}^{\frac{\pi i}{2}(a-b)}\frac{\Gamma(2a)}{\Gamma(a+b)}\phi(-a+b,1-2a,\zeta)\\
        \mathrm{e}^{-\frac{\pi i}{2}(a-b)}\frac{\Gamma(1+a+b)}{\Gamma(1+2a)}\phi(1+a+b,1+2a,\zeta) & \mathrm{e}^{\frac{\pi i}{2}(a+b)}\frac{\Gamma(2a)}{\Gamma(a-b)}\phi(1-a+b,1-2a,\zeta)
    \end{pmatrix}\mathrm{e}^{-\frac{1}{2}a\pi i\sigma_3}
    \end{split}
\end{equation}
\end{small}
for $\zeta\in\Omega_{Z,1},~ 2a\notin\mathbb{N}$, and $\phi(a,b,\zeta)$ is the confluent hypergeometric function defined as follows \begin{equation}\label{eq:defchf}
    \phi(a,b,\zeta)=1+\sum_{k=1}^{\infty}\frac{a(a+1)...(a+k-1)\zeta^k}{b(b+1)...(b+k-1)k!}.
\end{equation}
The expression of $Z(\zeta)$ in the other regions is then determined by using \eqref{eq:localZ0}, \eqref{eq:defZ0} and the jump condition \eqref{eq:jumpZ}. The case $2a\in\mathbb{N}$ can be constructed in a similar way; cf. \cite{DZ22}.

\end{description}

With a selection of the following parameters in $Z(\zeta;a,b)$,
\begin{equation}
    a=-\alpha,~ b=-\mu, 
\end{equation}
then the RH problem for $P^{(0)}(\lambda)$ can then be constructed by
\begin{equation}\label{eq:solutionSl}
    P^{(0)}(\lambda)=P^{(\infty)}(\lambda)(x\lambda)^{-\mu\sigma_3}\mathrm{e}^{\frac{3}{2}\pi i\mu\sigma_3}\mathrm{e}^{-\xi\sigma_3}Z(x\lambda;-\alpha,-\mu)\mathrm{e}^{\xi\sigma_3}\mathrm{e}^{\frac{1}{2}\lambda x\sigma_3},
\end{equation}
where 
\begin{equation}\label{eq:e2beta}
    \mathrm{e}^{2\xi}=\frac{\mathrm{e}^{2\pi i\mu}-\mathrm{e}^{-2\pi i\alpha}}{m_{12}\mathrm{e}^{-\pi i\alpha}},
\end{equation}
with $m_{12}$ given in \eqref{eq:def M}.
From \eqref{eq:defPinfty} and \eqref{eq:Z at infity}, as $ix\to-\infty$, we have the following expansion
\begin{equation}\label{eq:asySl}
    P^{(0)}(\lambda)=\left(I-\frac{P_0(\lambda)}{ix}+\mathcal{O}(x^{-2})\right)P^{(\infty)}(\lambda),~\lambda\in U(0)\cap\Omega_1,
\end{equation}
where
\begin{equation}\label{eq:defM1}
    P_0(\lambda)=\frac{1}{\lambda}\mathrm{e}^{-\frac{\pi}{2}i}(\lambda-1)^{-\mu\sigma_3}(-ix)^{-\mu\sigma_3}\mathrm{e}^{\pi i\mu\sigma_3}\mathrm{e}^{-\xi\sigma_3}Z_1(-\alpha,-\mu)\mathrm{e}^{\xi\sigma_3}\mathrm{e}^{-\pi i\mu\sigma_3}(-ix)^{\mu\sigma_3}(\lambda-1)^{\mu\sigma_3},
\end{equation}
with $-ix>0$ and the  matrix $Z_1(a,b)$ defined in \eqref{eq:Z1}.
\subsection{Local parametrix at $\lambda=1$}
It remains to construct the local parametrix near $\lambda=1$. We seek a function $P^{(1)}(\lambda)$ solving a RH problem in the neighborhood of $U(1)=\lbrace \lambda: |\lambda-1|<\delta\rbrace$ with $\delta<\frac{1}{2}$ as follows.
\subsubsection*{RH problem for $P^{(1)}$}
\begin{description}
    \item(1)
    $P^{(1)}(\lambda)$ is analytic in $U(1)\backslash\Sigma_T$ .
\item (2)
$P^{(1)}(\lambda)$ has the same jump condition as $T(\lambda)$ for $\lambda\in\Sigma_T\cap U(1).$ 
\item (3)
As $ix\to -\infty$, we have the matching 
condition 
\begin{equation}
    P^{(1)}(\lambda)=\left(I+\mathcal{O}(x^{-1})\right)P^{(\infty)}(\lambda),
\end{equation}
where the error term is uniform for $\lambda\in \partial U(1).$
\item (4)
$P^{(1)}(\lambda)$ remains bounded in a neighborhood of $\lambda=1$.
\end{description}

We choose $a=\alpha$ and $b=-\mu$ in the confluent hypergeometric parametrix $Z(\zeta;a,b)$. Consequently, we obtain the explicitly solution for $P^{(1)}(\lambda)$, as follows:
\begin{equation}\label{eq:solutionSr2}
   P^{(1)}(\lambda)=P^{(\infty)}(\lambda)(-x (\lambda-1))^{\mu\sigma_3}\mathrm{e}^{-\frac{1}{2}x\sigma_3}\mathrm{e}^{-\xi\sigma_3}\sigma_1\mathrm{e}^{-\frac{1}{2}\pi i\mu\sigma_3}Z(-x (\lambda-1);\alpha,-\mu)\sigma_1\mathrm{e}^{\xi\sigma_3}\mathrm{e}^{\frac{1}{2}\lambda x\sigma_3}.
\end{equation}
Then, we have the expansion
\begin{equation}\label{eq:asy P1}
    P^{(1)}(\lambda)=\left(I-\frac{P_1(\lambda)}{ix}+\mathcal{O}(x^{-2})\right)P^{(\infty)}(\lambda)
\end{equation}
as $ix\to -\infty$ uniformly for $\lambda\in U(1)\cap\Omega_1$. Here we have
\begin{equation}\label{eq:defM2}
    P_1(\lambda)=\frac{1}{\lambda-1}\mathrm{e}^{\frac{1}{2}\pi i}\lambda^{\mu\sigma_3}\mathrm{e}^{-\frac{1}{2}x\sigma_3}(-ix)^{\mu\sigma_3}\mathrm{e}^{-\xi\sigma_3}\sigma_1Z_1(\alpha,-\mu)\sigma_1\mathrm{e}^{\xi\sigma_3}\mathrm{e}^{\frac{1}{2}x\sigma_3}(-ix)^{-\mu\sigma_3}\lambda^{-\mu\sigma_3},
\end{equation}
with $-ix>0$ and the  matrix $Z_1(a,b)$ defined in \eqref{eq:Z1}.

\subsection{Final transformation}
In the final transformation, we define 
\begin{equation}\label{eq:def R}
    R(\lambda)=\begin{cases}
        T(\lambda)P^{(1)}(\lambda)^{-1}, & \lambda\in U(1)\backslash\Sigma_T,\\
        T(\lambda)P^{(0)}(\lambda)^{-1}, & \lambda\in U(0)\backslash\Sigma_T,\\
        T(\lambda)P^{(\infty)}(\lambda)^{-1}& eleswhere.
    \end{cases}
\end{equation}
Then $R(\lambda)$ satisfies the following RH problem.
\subsubsection*{RH problem for $R$}
\begin{description}
    \item(1)
    $R(\lambda)$ is analytic for $\lambda \in \mathbb{C}\backslash\Sigma_R$, where the contours are shown in Figure \ref{fig:fig4}.
    \begin{figure}[h]
    \centering
    \includegraphics[width=0.5\textwidth]{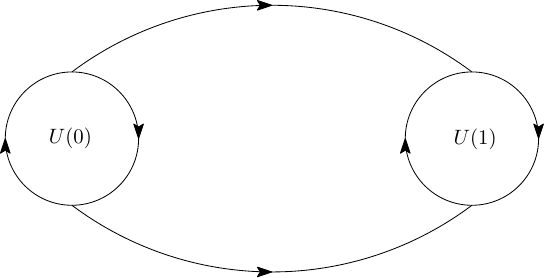}
    \caption{The contours of the RH problem for $R$.}
    \label{fig:fig4}
\end{figure}
\item (2)
$R(\lambda)$ satisfies the jump condition 
\begin{equation}
    R_+(\lambda)=R_-(\lambda)J_R(\lambda),
\end{equation}
where
\begin{equation}
    J_R(\lambda)=\begin{cases}
        P^{(1)}(\lambda)P^{(\infty)}(\lambda)^{-1}, & \lambda \in \partial U(1),\\
        P^{(0)}(\lambda)P^{(\infty)}(\lambda)^{-1}, & \lambda \in \partial U(0),\\
        P^{(\infty)}(\lambda)J_T(\lambda)P^{(\infty)}(\lambda)^{-1} & elsewhere.
    \end{cases}
\end{equation}
\item (3)
As $\lambda \to \infty$, we have
\begin{equation}
    R(\lambda)=I+\frac{R_1(x)}{\lambda}+\mathcal{O}(\lambda^{-2}).
\end{equation}
\end{description}
It follows from \eqref{eq:defJT} and \eqref{eq:defPinfty} that the jump matrix $P^{(\infty)}(\lambda)J_T(\lambda)P^{(\infty)}(\lambda)^{-1}$ tends to the identity matrix exponentially fast as $ix\to-\infty$, uniformly for $\lambda
\in\Sigma_R\backslash\lbrace\partial U(0)\cup\partial U(1)\rbrace$. For $\lambda\in\partial U(0)\cup\partial U(1)$, it follows from \eqref{eq:asySl} and \eqref{eq:asy P1} that
\begin{equation}\label{eq:asyJR2}
    J_R(\lambda)=I-\frac{J_{R,1}(\lambda)}{ix}+\mathcal{O}(x^{-2}), ~ix\to -\infty.
\end{equation}
This shows that the RH problem for $R(\lambda)$ is a small norm problem when $ix\to-\infty$, which gives us the following asymptotic approximation 
\begin{equation}\label{eq:asyR2}
    R(\lambda)=I-\frac{R^{(1)}(\lambda)}{ix}+\mathcal{O}(x^{-2})
\end{equation}
uniformly for $\lambda\in\mathbb{C}\backslash\Sigma_R$.

By \eqref{eq:asyJR2}, \eqref{eq:asyR2} and Plemelj's formula, we have
\begin{equation}
    R^{(1)}(\lambda)=\frac{1}{2\pi i}\int_{\partial U(0)}\frac{J_{R,1}(\zeta)}{\zeta-\lambda}d\zeta+\frac{1}{2\pi i}\int_{\partial U(1)}\frac{J_{R,1}(\zeta)}{\zeta-\lambda}d\zeta,
\end{equation}
where 
\begin{equation}\label{eq:defJR1}
    J_{R,1}(\lambda)=\begin{cases}
        P_0(\lambda), &\lambda\in\partial U(0),\\
        P_1(\lambda), &\lambda\in\partial U(1).
    \end{cases}
\end{equation}
For later use, we derive the estimate for $R^{(1)}(\lambda)$ below
\begin{equation}\label{eq:intR1}
    R^{(1)}(\lambda)=\frac{1}{\lambda}\left(-\frac{1}{2\pi i}\int_{\partial U(0)}J_{R,1}(\zeta)d\zeta-\frac{1}{2\pi i}\int_{\partial U(1)}J_{R,1}(\zeta)d\zeta\right)+\mathcal{O}(\lambda^{-2}).
\end{equation}
A combination of \eqref{eq:defM1}, \eqref{eq:defM2}, \eqref{eq:defJR1} and \eqref{eq:intR1} yields the following estimate for $R_1(x)$
\begin{equation}\label{eq:hatR1}
    R_1(x)=\frac{\mu^2-\alpha^2}{x}\begin{pmatrix}
        -2 & |\frac{\Gamma(\mu-\alpha)}{\Gamma(-\alpha-\mu)}|\frac{2i}{\alpha+\mu} \mathrm{e}^{-2\xi-\frac{1}{2}x}\sin{(d(x))}\\
        |\frac{\Gamma(\mu+\alpha)}{\Gamma(\alpha-\mu)}|\frac{2i}{\alpha-\mu} \mathrm{e}^{2\xi+\frac{1}{2}x}\sin{(d(x))} & 2
    \end{pmatrix}+\mathcal{O}(x^{-2}),
\end{equation}
where
\begin{equation}\label{eq:def d(x)}
    d(x)=-\frac{ix}{2}+2\mu i\ln(-ix)+\arg{\Gamma(\mu-\alpha)}-\arg{\Gamma(-\alpha-\mu)}.
\end{equation}

\subsection{Proof of Theorem \ref{thm:2}: asymptotics of the PV transcendents as $ix\to-\infty$}\label{subsec:proof large x}
Using the results of the nonlinear steepest descent analysis of the RH problem for $\Phi$, we obtain the asymptotic expansions for one parameter family of solutions to
the PV equation as $ix\to -\infty$. Tracking back the series of transformations performed in Section \ref{sec:asy of PV for large x}  
\begin{equation*}
    \Phi\mapsto Y\mapsto T\mapsto R,
\end{equation*}
we have for large $\lambda$,
\begin{equation}
    \Phi(\lambda)=R(\lambda)P^{(\infty)}(\lambda)\mathrm{e}^{-\frac{1}{2}\lambda x\sigma_3}.
\end{equation}
This leads to
\begin{equation}\label{eq:defY-1}
     \Phi_{-1}(x)=R_1(x)+\mu\sigma_3.
\end{equation}
Combining \eqref{eq:hatR1} and \eqref{eq:defY-1}, we derive that as $ix\to -\infty$,
\begin{align}
    (\Phi_{-1})_{11}=\mu+\frac{2}{x}(\alpha^2-\mu^2)+\mathcal{O}(x^{-2}), \\    
    (\Phi_{-1})_{12}(\Phi_{-1})_{21}=\frac{4(\mu^2-\alpha^2)}{x^2}\sin^2{(d(x))}+\mathcal{O}(x^{-3}).
\end{align}
Using \eqref{eq:RH for v2} and \eqref{eq:RH Ham for PV}, we then obtain the asymptotic expansions of $v(x)$ and $H_V(x)$ given in \eqref{eq:asyofv} and \eqref{eq:asyofHV} as $ix\to -\infty$. 

In order to derive the asymptotic behavior of $u(x)$, $y(x)$ and $\hat{u}(x)$, we use the local behavior of $T(\lambda)$ near $\lambda=0$ and $\lambda=1$. From \eqref{eq:Tlocal0} and \eqref{eq:solutionSl}, it follows that as $\lambda\to 0$,
\begin{equation}    T(\lambda)={\Phi}^{(0)}_0(x)\lambda^{-\alpha\sigma_3}\mathrm{e}^{\pi i\alpha\sigma_3}U^{-1}S_L^{-1}, ~\lambda\in\Omega_1,
\end{equation}
\begin{equation}
     P^{(0)}(\lambda)=\mathrm{e}^{-\xi\sigma_3}(-ix)^{-\mu\sigma_3}\mathrm{e}^{-\frac{\pi i}{2}\mu\sigma_3}Q(0) (-ix)^{-\alpha\sigma_3}\lambda^{-\alpha\sigma_3}B_1B_2^{-1}\mathrm{e}^{-\frac{\pi i}{2}\mu\sigma_3}\mathrm{e}^{\xi\sigma_3},
\end{equation}
where 
\begin{equation}
    Q(0)=\begin{pmatrix}
        \mathrm{e}^{\frac{\pi i}{2}(\alpha+\mu)}\frac{\Gamma(1+\mu-\alpha)}{\Gamma(1-2\alpha)} &  -\mathrm{e}^{\frac{\pi i}{2}(-\alpha+\mu)}\frac{\Gamma(-2\alpha)}{\Gamma(-\alpha-\mu)}\\
         \mathrm{e}^{\frac{\pi i}{2}(\alpha-\mu)}\frac{\Gamma(1-\mu-\alpha)}{\Gamma(1-2\alpha)} & \mathrm{e}^{-\frac{\pi i}{2}(\alpha+\mu)}\frac{\Gamma(-2\alpha)}{\Gamma(-\alpha+\mu)}
    \end{pmatrix},
\end{equation}
and $B_1$, $B_2$ are defined in \eqref{eq:defC1}.
By using \eqref{eq:def R} and \eqref{eq:asyR2}, we can calculate that as $ix\to-\infty$,
\begin{equation}\label{eq:asy Phi0}
    {\Phi}^{(0)}_0(x)=\mathrm{e}^{-\xi\sigma_3}(-ix)^{-\mu\sigma_3}\mathrm{e}^{-\frac{\pi i}{2}\mu\sigma_3}Q(0) (ix)^{-\alpha\sigma_3}f_1^{\sigma_3}+\mathcal{O}(x^{-1}),
\end{equation}
where $f_1=\omega\mathrm{e}^{\xi-\frac{\pi i}{2}\mu}$.

From \eqref{eq:asy Phi0} and \eqref{eq:analyticfac}, we  obtain the asymptotic expansions of $a_0(x),~ b_1(x),~y(x),~ v(x)$ and $u(x)$ as $ix\to-\infty$. Specifically, we have the following asymptotic expansions 
\begin{align}
    a_0\frac{y}{v}(v-2\alpha)=&\mathrm{e}^{-\xi+\frac{\pi i}{2}\alpha}(-ix)^{-\mu}(ix)^{-\alpha}\frac{\Gamma(1+\mu-\alpha)}{\Gamma(1-2\alpha)}f_1+\mathcal{O}(x^{-1}),\label{eq:a0y(v-2al)/y}\\
    b_0y=&-\mathrm{e}^{-\xi-\frac{\pi i}{2}\alpha}(-ix)^{-\mu}(ix)^{\alpha}\frac{\Gamma(-2\alpha)}{\Gamma(-\alpha-\mu)}f_1^{-1}+\mathcal{O}(x^{-1}),\label{eq;b0y}\\
    a_0=&\mathrm{e}^{\xi+\frac{\pi i}{2}\alpha}(-ix)^{\mu}(ix)^{-\alpha}\frac{\Gamma(1-\mu-\alpha)}{\Gamma(1-2\alpha)}f_1+\mathcal{O}(x^{-1}),\label{eq:a0}\\
    b_0=&\mathrm{e}^{\xi-\frac{\pi i}{2}\alpha}(-ix)^{\mu}(ix)^{\alpha}\frac{\Gamma(-2\alpha)}{\Gamma(\mu-\alpha)}f_1^{-1}+\mathcal{O}(x^{-1}).\label{eq:b0}
\end{align}

The derivation of ${\Phi}^{(1)}_0(x)$ is analogous to that of ${\Phi}^{(0)}_0(x)$, from \eqref{eq:Tlocal1} and \eqref{eq:solutionSr2} we obtain that as $\lambda\to 1$,
\begin{equation} T(\lambda)={\Phi}^{(1)}_0(x)(\lambda-1)^{\alpha\sigma_3}U^{-1}S_L^{-1}\mathrm{e}^{\frac{1}{2} x\sigma_3}, ~\lambda\in\Omega_1.
\end{equation}
\begin{equation}
\begin{small}
    P^{(1)}(\lambda)=(-ix)^{\mu\sigma_3}\mathrm{e}^{-\xi\sigma_3}\mathrm{e}^{-\frac{1}{2}x\sigma_3}\sigma_1\mathrm{e}^{-\frac{\pi i}{2}\mu\sigma_3}Q(1)\mathrm{e}^{-\pi i\alpha\sigma_3} (-ix)^{\alpha\sigma_3}(\lambda-1)^{\alpha\sigma_3}B_1B_2^{-1}\mathrm{e}^{\frac{3\pi i}{2}\mu\sigma_3}\sigma_1\mathrm{e}^{\xi\sigma_3}\mathrm{e}^{\frac{1}{2}x\sigma_3}.
\end{small}
\end{equation}
where
\begin{equation}
    Q(1)=\begin{pmatrix}
        \mathrm{e}^{\frac{\pi i}{2}(-\alpha+\mu)}\frac{\Gamma(1+\mu+\alpha)}{\Gamma(1+2\alpha)} &  -\mathrm{e}^{\frac{\pi i}{2}(\alpha+\mu)}\frac{\Gamma(2\alpha)}{\Gamma(\alpha-\mu)}\\
         \mathrm{e}^{\frac{\pi i}{2}(-\alpha-\mu)}\frac{\Gamma(1-\mu+\alpha)}{\Gamma(1+2\alpha)} & \mathrm{e}^{\frac{\pi i}{2}(\alpha-\mu)}\frac{\Gamma(2\alpha)}{\Gamma(\alpha+\mu)}  
    \end{pmatrix}.
\end{equation}
With a similar computation, we have as $ix\to-\infty$, 
\begin{equation}    
{\Phi}^{(1)}_0(x)=(-ix)^{\mu\sigma_3}\mathrm{e}^{-\xi\sigma_3}\mathrm{e}^{-\frac{1}{2}x\sigma_3}\sigma_1\mathrm{e}^{-\frac{\pi i}{2}\mu\sigma_3}Q(1) (-ix)^{\alpha\sigma_3}\mathrm{e}^{-\pi i\alpha\sigma_3}\begin{pmatrix}
    f_2 & 0\\
    0 & -f_2^{-1}
\end{pmatrix}+\mathcal{O}(x^{-1}),
\end{equation}
where $f_2=\mathrm{e}^{(\xi+\pi i\alpha-\frac{\pi i}{2}\mu)}\omega$. Therefore, we have the following asymptotic expansions as $ix\to-\infty$, 
\begin{align}
    a_1=&\mathrm{e}^{-\xi-\frac{1}{2}x-\frac{3\pi i}{2}\alpha}(-ix)^{\alpha+\mu}\frac{\Gamma(1-\mu+\alpha)}{\Gamma(1+2\alpha)}f_2+\mathcal{O}(x^{-1}),\label{eq:a1}\\
    b_1yu=&-\mathrm{e}^{-\xi-\frac{1}{2}x+\frac{3\pi i}{2}\alpha}(-ix)^{-\alpha+\mu}\frac{\Gamma(2\alpha)}{\Gamma(\mu+\alpha)}f_2^{-1}+\mathcal{O}(x^{-1}),\label{eq:b1yu}\\
    \frac{a_1}{yu}\frac{v}{v-2\alpha}=&\mathrm{e}^{\xi+\frac{1}{2}x-\frac{3\pi i}{2}\alpha}(-ix)^{\alpha-\mu}\frac{\Gamma(1+\mu+\alpha)}{\Gamma(1+2\alpha)}f_2+\mathcal{O}(x^{-1}),\label{eq:a1v/yu(v-2al)}\\
    b_1=&\mathrm{e}^{\xi+\frac{1}{2}x+\frac{3\pi i}{2}\alpha}(-ix)^{-\alpha-\mu}\frac{\Gamma(2\alpha)}{\Gamma(\alpha-\mu)}f_2^{-1}+\mathcal{O}(x^{-1}).\label{eq:b1}
\end{align}
Combining \eqref{eq:a0y(v-2al)/y}-\eqref{eq:b0}, \eqref{eq:a1}-\eqref{eq:b1}, \eqref{eq:e2beta}
and \eqref{eq:def M}, we obtain \eqref{eq:asyofu}, \eqref{eq:asyofy} and \eqref{eq:asyofhu}. This completes the proof of the large-$x$ asymptotic expansions of Theorem \ref{thm:2}.

\section{Asymptotic analysis of the RH problem for $\Phi$ as $x\to0$}\label{sec:asy of PV for small x}
In this section, we will calculate the asymptotic expansions of the PV transcendents as $x\to0$. From Proposition \ref{pro:pro1}, we observe that the RH problem for $\Phi(\lambda,x)$ is pole-free for $x\in i\mathbb{R}$. Thus, we can derive the leading asymptotic behavior of the PV transcendents as $x\to0$ by finding the exact solution of the RH problem for $\Phi(\lambda,x)$ at $x=0$.
\subsection{The solution of The RH problem for $\Phi$ at $x=0$}
Let $x=0$, we obtain the following RH problem for $\Phi(\lambda,0)$.
\subsubsection*{RH problem for $\Phi(\lambda,0)$}
\begin{description}
    \item(1)
    $\Phi(\lambda,0)$ is analytic for $\lambda\in\mathbb{C}\backslash[0,1]$.
\item (2)
$\Phi(\lambda,0)$ satisfies the jump condition
\begin{equation}
\Phi_+(\lambda,0)=\Phi_-(\lambda,0)U\mathrm{e}^{2\pi i\alpha\sigma_3}U^{-1}, ~\lambda\in(0,1).
\end{equation}
\item (3)
The asymptotic behavior of $\Phi(\lambda,0)$ at infinity is
\begin{equation}
\Phi(\lambda,0)=I+\frac{\Phi_{-1}(0)}{\lambda}+\mathcal{O}(\lambda^{-2}).
\end{equation}
\item (4)
The behavior of $\Phi(\lambda,0)$ at $\lambda=0$ is
\begin{equation}
\Phi(\lambda,0)={\Phi}^{(0)}(\lambda)\lambda^{-\alpha\sigma_3}\mathrm{e}^{\pi i\alpha\sigma_3}U^{-1}, ~ \lambda\to 0.
\end{equation}
\item (5)
The behavior of $\Phi(\lambda,0)$ at $\lambda=1$ is
\begin{equation}
\Phi(\lambda,0)={\Phi}^{(1)}(\lambda)(\lambda-1)^{\alpha\sigma_3}U^{-1},  ~\lambda\to 1.
\end{equation}
\end{description}
Then $\Phi(\lambda,0)$ can be solved explicitly by
\begin{equation}\label{eq:solution of Phi at x=0}
    \Phi(\lambda,0)=U\left(\frac{\lambda}{\lambda-1}\right)^{-\alpha\sigma_3}U^{-1},
\end{equation}
where the function $\left(\frac{\lambda}{\lambda-1}\right)^{\alpha}$ takes the branch cut along $[0,1]$, with its branch fixed by the asymptotic condition
\begin{equation*}
    \left(\frac{\lambda}{\lambda-1}\right)^\alpha \to 1,~\lambda\to \infty.
\end{equation*}

\subsection{Proof of Theorem \ref{thm:2}: asymptotics of the PV transcendents as $x\to 0$ }
From \eqref{eq:solution of Phi at x=0}, we find that 
\begin{equation}
    \Phi_{-1}(0)=-\alpha U\sigma_3U^{-1}.
\end{equation}
This, together with Proposition \ref{pro:pro1},  allows us to derive the Taylor expansion of the PV transcendents at $x=0$:
\begin{equation}
H_V(x)=-\alpha(2|\omega|^2-1)+\mathcal{O}(x),~~~x\to 0,
\end{equation}
\begin{equation}
v(x)=2\alpha(1-|\omega|^2))+\mathcal{O}(x),~~~x\to 0,
\end{equation}
\begin{equation}\label{eq:analytic factor at x=i0}
    {\Phi}^{(1)}_0(x)=U+\mathcal{O}(x),~{\Phi}^{(0)}_0(x)=U+\mathcal{O}(x),~~~x\to 0.
\end{equation}
From \eqref{eq:analytic factor at x=i0}, we
conclude that
\begin{equation}
    u(x)=1+\mathcal{O}(x),~~~x\to 0,
\end{equation}
\begin{equation}
    y(x)=\frac{\sqrt{1-|\omega|^2}}{\overline{\omega}}+\mathcal{O}(x),~\omega\neq 0,~~~x\to 0,
\end{equation}
\begin{equation}
    \hat{u}(x)=1+\mathcal{O}(x),~~~x\to 0.
\end{equation}
This completes the proof of the small-$x$ asymptotic expansions of Theorem \ref{thm:2}. This, together with the large-$x$ asymptotic formulas obtained in Section \ref{subsec:proof large x}, the existence results given in Proposition \ref{pro:pro1} and the properties in Proposition \ref{pro:2}, completes the proof of Theorem \ref{thm:2}.

\section{Asymptotic analysis of the RH problem for $\Psi$ as $s\to-\infty$}\label{sec: large t}

In this section, we perform the Deift-Zhou nonlinear steepest descent analysis \cite{DIZ, DMVZ1, DMVZ2, DZ} of the RH problem for $\Psi(\lambda,\Vec{s})$ as $s\to-\infty$, where the variables $\Vec{s}=(s_1,s_2)$ are given by
\begin{equation}
s_1=s+\frac{\tau}{\sqrt{-s}},~ s_2=s-\frac{\tau}{\sqrt{-s}}, ~\tau \ge 0.
\end{equation}

\subsection{Normalization and deformations of the jump curves}
To normalize the asymptotic behavior of $\Psi(\lambda)$ at infinity, we introduce the $g$-function
\begin{equation}\label{eq:def g}
g(z)=\frac{z^3}{3}-2z.
\end{equation}
It is direct to see that $g(z)$ has two saddle points, namely, the points satisfying $g'(z)=0$,
\begin{equation}
z_+=\sqrt{2},~ z_-=-\sqrt{2}.
\end{equation}

We introduce the first transformation
\begin{equation}
Y(z)=\Psi(\sqrt{-s}z)\mathrm{e}^{\frac{itg(z)}{2}\hat{\sigma}_3},
\end{equation}
with $t=(-s)^{\frac{3}{2}}$. As a result, we have 
\begin{equation}
\Theta(\sqrt{-s}\lambda)=\frac{it}{2}g(z)\hat{\sigma}_3+\tau iz\sigma_3\otimes\sigma_3.
\end{equation}
Then $Y(z)$ satisfies the following RH problem.
\subsubsection*{RH problem for $Y$}
\begin{description}
    \item(1)
    $Y(z)$ is a $4\times4$ matrix valued function analytic in $\mathbb{C}\backslash\Sigma_k,~ k=1,2,3,4$; see Figure \ref{fig:fig1} for the contours.
\item (2)
$Y(z)$ satisfies the jump condition 
\begin{equation}
Y_+(z) = Y_-(z)\hat{S_k}(z), ~z\in\Sigma_k, ~k=1,2,3,4,
\end{equation}
where
\begin{equation}
\begin{split}
\hat{S_1}(z)=\begin{pmatrix}I & 0\\
\mathrm{e}^{itg(z)}C & I\end{pmatrix},
~~\hat{S_2}(z)=\begin{pmatrix}I & 0\\
-\mathrm{e}^{itg(z)}C & I\end{pmatrix},\\
\hat{S_3}(z)=\begin{pmatrix}I & \mathrm{e}^{-itg(z)}C\\
0 & I\end{pmatrix},
~~\hat{S_4}(z)=\begin{pmatrix}I & -\mathrm{e}^{-itg(z)}C\\
0 & I\end{pmatrix},
\end{split}
\end{equation}
with $t=(-s)^{\frac{3}{2}}$.
\item (3)
As $z\to\infty$, we have
\begin{equation}\label{eq:asyPsi}
Y(z) = \left(I_4+\frac{\Psi_1(\Vec{s})}{z}+\mathcal{O}(z^{-2})\right)\mathrm{e}^{-\tau iz\sigma_3\otimes\sigma_3}.
\end{equation}
\end{description}

\subsection{Opening the lenses and deformation}
Next, we transform the above RH problem into a RH problem formulated on the anti-Stokes curves of $g(z)$ by using a similar process in \cite[Chapter 9]{FIKY}. We observe that $Y(z)$ can be expressed as a RH problem posed on the curves depicted in Figure \ref{fig:fig5}, where we use the notation $\hat{S_k}, k=1,2,3,4$ to denote the corresponding jump matrices.
It is observed that the jump matrix on $[-\sqrt{2},\sqrt{2}]$ are highly oscillating for large $t$. To eliminate these oscillations, we turn them into exponential small terms on the anti-Stokes curves of $g(z)$. Specifically, we can factorize the jump matrix on $[-\sqrt{2},\sqrt{2}]$ as follows:
\begin{equation}
   \hat{S_4}(z)\hat{S_1}(z)=S_L(z)S_DS_U(z),
\end{equation}
where
\begin{equation}\label{eq:def SU,SL,SD}
\begin{split}
&S_L(z)=\begin{pmatrix}
    I & 0\\
    \mathrm{e}^{itg(z)}C(I-C^2)^{-1} & I
\end{pmatrix},
~~~S_D=\begin{pmatrix}
    I-C^2 & 0\\
    0 & (I-C^2)^{-1}  
\end{pmatrix},\\
&S_U(z)=\begin{pmatrix}
    I & -\mathrm{e}^{-itg(z)}(I-C^2)^{-1}C\\
    0 & I
\end{pmatrix}.
\end{split}
\end{equation}
\begin{figure}
    \centering
    \includegraphics[width = 0.7\textwidth]{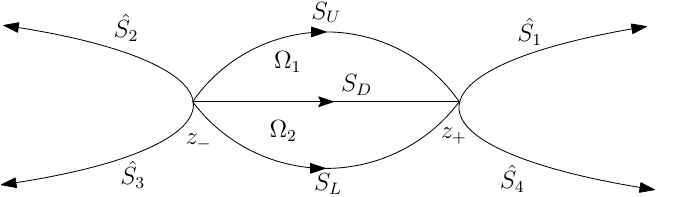}
    \caption{The deformation of the jump contours of the original RH problem.}
    \label{fig:fig5}
\end{figure}

Based on the factorization of jumps and deformation of curves, we introduce the second transformation
\begin{equation} 
S(z)=Y(z)\begin{cases}
S_U^{-1}(z),& z\in\Omega_1,\\
S_L(z),& z\in\Omega_2,\\
I& elsewhere.
\end{cases}
\end{equation}
After blowing up the lens-shaped regions $\Omega_1$ and $\Omega_2$, we finally arrive at the following  RH problem for $S(z)$, formulated on the anti-Stokes curves of $g(z)$ as shown in Figure \ref{fig:fig6}.
\subsubsection*{RH problem for $S$}
\begin{description}
    \item(1)
    $S(z)$ is a $4\times4$ matrix valued function  analytic in $\mathbb{C}\backslash\Sigma_S$, where the contours are shown in Figure \ref{fig:fig6}.
    \begin{figure}[h]
    \centering
    \includegraphics[width = 0.6\textwidth]{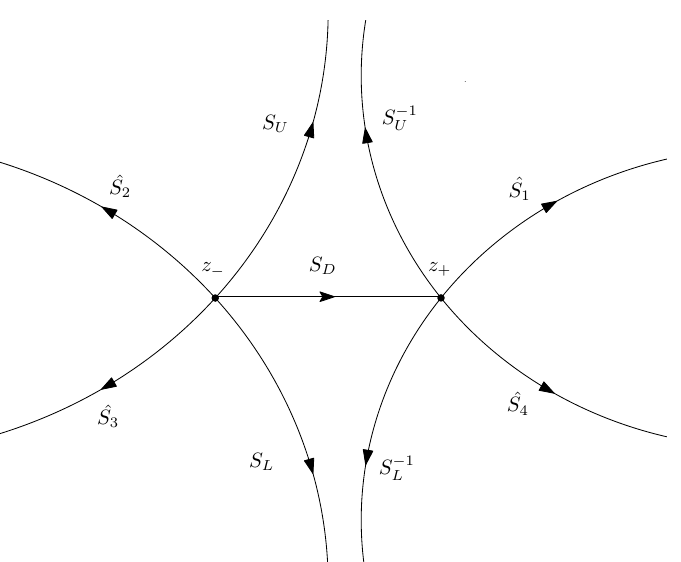}
    \caption{The contours $\Sigma_S$ of the RH problem for $S$.}
    \label{fig:fig6}
\end{figure}
\item (2)
$S(z)$ satisfies the jump condition 
\begin{equation}
S_+(z) = S_-(z)J_S(z),  
\end{equation}
where
\begin{equation}\label{eq:defJs}
J_S(z)=\begin{cases}
\hat{S_k},& z\in\Sigma_{S,k},k=1,4,5,8,\\
S_U^{-1}(z),& z\in\Sigma_{S,2},\\
S_U(z),& z\in\Sigma_{S,3},\\
S_L(z),& z\in\Sigma_{S,6},\\
S_L^{-1}(z),& z\in\Sigma_{S,7},\\
S_D,& z\in\Sigma_{S,9,}
\end{cases}
\end{equation}
with $S_L$, $S_D$ and $S_U$ defined in \eqref{eq:def SU,SL,SD}.
\item (3)
As $z\to\infty$, we have
\begin{equation}
S(z) = (I_4+\mathcal{O}(z^{-1}))\mathrm{e}^{-\tau iz\sigma_3\otimes\sigma_3}.
\end{equation}
\end{description}
\subsection{Global parametrix}
The jump matrices for $S(z)$ tend to the identity matrix exponentially fast, except for those on the real axis as $t\to +\infty$. Therefore, we arrive at the solution to the following approximation RH problem for a $4\times4$ matrix-valued function $S^{(\infty)}(z)$.
\subsubsection*{RH problem for $S^{(\infty)}$}
\begin{description}
    \item(1)
    $S^{(\infty)}(z)$ is a $4\times4$ matrix valued function  analytic for $z\in\mathbb{C}\backslash[-\sqrt{2},\sqrt{2}]$.
\item (2)
$S^{(\infty)}(z)$ satisfies the following jump condition
\begin{equation}\label{eq:jump Pinfty}
S^{(\infty)}_+(z)=S^{(\infty)}_-(z)S_D,~z\in(-\sqrt{2},\sqrt{2}),
\end{equation}
with $S_D$ given in \eqref{eq:def SU,SL,SD}.
\item (3)
As $z\to\infty$, we have
\begin{equation}\label{eq:asyPinfty}
    S^{(\infty)}(z)=\left(I_4+\frac{S^{(\infty)}_1}{z}+\mathcal{O}(z^{-2})\right)\mathrm{e}^{-\tau iz\sigma_3\otimes\sigma_3}.
\end{equation}
\end{description}
To solve this RH problem, we first define the following auxiliary function $f(z)$ that
\begin{equation}\label{eq:solution of f}
f(z)=\left(\frac{z+\sqrt{2}}{z-\sqrt{2}}\right)^{-\frac{1}{4\pi i}\ln[(1-\lambda_1^2)(1-\lambda_2^2)]},
\end{equation}
where $f(z)$ takes the branch cut along $[-\sqrt{2},\sqrt{2}]$ and has the following behavior as $z\to\infty$
\begin{equation}
f(z)=1+\mathcal{O}(z^{-1}),~z\to\infty.
\end{equation}
From the definition of $f(z)$, it follows that $f(z)$ is analytic in $\mathbb{C}\backslash[-\sqrt{2},\sqrt{2}]$, with jump condition
as follows:
\begin{equation}
f_+(z)=f_+(z)[(1-\lambda_1^2)(1-\lambda_2^2)]^{\frac{1}{2}},~z\in(-\sqrt{2},\sqrt{2}).
\end{equation}

And then, we introduce the transformation
\begin{equation}\label{eq:relation N and Pinfty}
    N(z)=S^{(\infty)}(z)\times \mathrm{diag}(f^{-1}(z),f^{-1}(z),f(z),f(z)).
\end{equation}
We see that $N(z)$ is analytic in $\mathbb{C}\backslash[-\sqrt{2},\sqrt{2}]$ and satisfies the following jump condition
\begin{equation}\label{eq:jump N}
    N_+(z)=N_-(z)\begin{pmatrix}
       U\mathrm{e}^{2\pi i\alpha\sigma_3}U^{-1} & 0\\
       0 & U\mathrm{e}^{-2\pi i\alpha\sigma_3}U^{-1}
    \end{pmatrix},~z\in(-\sqrt{2},\sqrt{2}).
\end{equation}
Due to the jump condition \eqref{eq:jump N} and the asymptotic behavior at infinity \eqref{eq:asyPinfty} for $N(z)$ are both block diagonal matrices, we can  
divide it into different  $2\times 2$ matrix blocks and construct the solution for $N(z)$ as follows
\begin{equation}\label{eq:solution N}
    N(z)=\begin{pmatrix}
        N_1(z) & 0\\
        0 & N_1(-z)
    \end{pmatrix},
\end{equation}
where 
\begin{equation}\label{eq:solution of N1}
    N_1(z)=\mathrm{e}^{\sqrt{2}\tau i\sigma_3}\Phi \left(\frac{z+\sqrt{2}}{2\sqrt{2}},4\sqrt{2}\tau i\right),\end{equation}
with $\Phi(\lambda,x)$ denoting the solution of the RH problem for the PV equation introduced in Section \ref{RHforFD and PV}.
\begin{remark}\label{re:deg solution of S}
    If $C^2$ is a diagonal matrix, we choose $U=I$ in \eqref{eq:def U}, then we have
    \begin{equation}
        S^{(\infty)}(z)=\left(\frac{z+\sqrt{2}}{z-\sqrt{2}}\right)^{\Lambda\otimes\sigma_3}\mathrm{e}^{-\tau iz\sigma_3\otimes\sigma_3}
    \end{equation}
    in this degenerate situation, with $\Lambda$ defined in \eqref{eq:def Lam}. 
\end{remark}

    For a general $r$, if we still choose $s_j=s+\frac{\tau_j}{\sqrt{-s}},~j=1,2,...,r$. Then we actually need to find solution to the following RH problem of a $2r\times 2r$ matrix-valued Global parametrix for $S^{(\infty)}(z)$:
    \begin{description}
        \item(1)
        $S^{(\infty)}(z)$ is a $2r\times2r$ matrix valued function  analytic for $z\in\mathbb{C}\backslash[-\sqrt{2},\sqrt{2}]$.
    \item(2)
    $S^{(\infty)}(z)$ satisfies the jump condition
    \begin{equation}
        S^{(\infty)}_+(z)=S^{(\infty)}_-(z)\begin{pmatrix}
            I_r-C^2 & 0\\
            0 & (I_r-C^2)^{-1}
        \end{pmatrix},~~z\in(-\sqrt{2},\sqrt{2}),
    \end{equation}
    here $C$ is a $r\times r$ Hermitian matrix.
    \item (3)
    As $z\to\infty$, we have
    \begin{equation}
         S^{(\infty)}(z)=\left(I_{2r}+\mathcal{O}(z^{-1})\right)\mathrm{e}^{-iz\boldsymbol{\tau}\otimes\sigma_3},
    \end{equation}
    where $\boldsymbol{\tau}\otimes\sigma_3=\begin{pmatrix}
        \boldsymbol{\tau} & 0\\
        0 & -\boldsymbol{\tau}
    \end{pmatrix}$, and $\boldsymbol{\tau}=\mathrm{diag}(\tau_1,\tau_2,...,\tau_r)$.
    \item (4)
    The local behaviors of $z=\pm\sqrt{2}$ are given by
    \begin{equation}
         S^{(\infty)}(z)=A(z)\left(\frac{z+\sqrt{2}}{z-\sqrt{2}}\right)^{\boldsymbol{\nu}\otimes\sigma_3}\begin{pmatrix}
             U^{-1} & 0\\
             0 & U^{-1}
         \end{pmatrix},
    \end{equation}
    where $U$ is a $r\times r$ unitary matrix and $\boldsymbol{\nu}=\mathrm{diag}(\nu_1,\nu_2,...,\nu_r)$ with $\nu_j=-\frac{1}{2\pi i}\ln(1-\lambda_j^2)$.
    \end{description}
    
\begin{remark}\label{re:NC PV}    If $r=1$, the RH problem for $S^{(\infty)}$ is reduced to that for the classical PV equation with special Stokes' multipliers. If $r=2$, as considered in this paper, then $S^{(\infty)}$ is a $4\times 4$ block diagonal matrix, with each block being the solution of the RH problem for the classical PV equation, as given in \eqref{eq:solution N} and \eqref{eq:solution of N1}.   For  general $r>2$,  we expect that the RH problem for $S^{(\infty)}$ corresponds to a special noncommutative (matrix-valued) PV equation. In fact, the matrix-valued Painlev\'e systems have been studied in recent years \cite{BGS, ber12, ber18, caf, K} etc. 
However, the related RH problem represention has not been addressed, to the best of our knowledge. In the next step, we will further study the RH problem for $S^{(\infty)}$, exploring the Lax integrability and its relation to  the noncommutative PV equation, 
    the existence and  uniqueness of the solution to the  RH problem for $S^{(\infty)}$ and the asymptotics of the  corresponding solutions of the noncommutative PV equation.  
\end{remark}
\subsection{Local parametrices}
Near the node points $z_-=-\sqrt{2}$ and $z_+=\sqrt{2}$, $S(z)$ can not be approximated by the global parametrix $S^{(\infty)}(z)$. Thus, in this section, we construct the local parametrices in the neighborhoods of $z_-$ and $z_+$.  
From \eqref{eq:def g}, we expand $g(z)$ in the neighborhood of $z=\sqrt{2}$ as follows:
\begin{equation}
    g(z)=g(\sqrt{2})+\frac{g''(\sqrt{2})}{2}(z-\sqrt{2})^{2}+\mathcal{O}((z-\sqrt{2})^3),
\end{equation}
where $g(\sqrt{2})=-\dfrac{4}{3}\sqrt{2},~ g''(\sqrt{2})=2\sqrt{2}.$

Firstly, we need to find a local parametrix 
$S^{(\sqrt{2})}(z)$ satisfying the
following RH problem.
\subsubsection*{RH problem for $S^{(\sqrt{2})}$}
\begin{description}
    \item(1)
    $S^{(\sqrt{2})}(z)$ is analytic in $U(\sqrt{2})\backslash\Sigma_S$, where $U(\sqrt{2})=\lbrace z: |z-\sqrt{2}|<\delta\rbrace$ with $\delta<\frac{1}{2}$. 
\item (2)
$S^{(\sqrt{2})}(z)$ has the same jump condition as $S(z)$ for $z\in\Sigma_S\cap U(\sqrt{2}).$ 
\item (3)
As $t\to +\infty$, we have the matching 
condition 
\begin{equation}\label{eq:matching1}
  S^{(\sqrt{2})}(z)=(I_4+\mathcal{O}(t^{-\frac{1}{2}}))S^{(\infty)}(z), ~z\in\partial U(\sqrt{2}),
\end{equation}
where the error term is uniform for $z\in \partial U(\sqrt{2}).$
\end{description}
We define the local conformal mapping 
\begin{equation}\label{expression for comformal variable}
    \zeta(z)=2\sqrt{\frac{i}{2}(g(\sqrt{2})-g(z))}.
\end{equation}
The branch of the square root \eqref{expression for comformal variable} is chosen in a way such that
\begin{equation}\label{eq:expansion for comformal}
    \zeta(z)=\zeta'(\sqrt{2})(z-\sqrt{2})+\frac{\zeta{''}(\sqrt{2})}{2}(z-\sqrt{2})^2+\mathcal{O}((z-\sqrt{2})^3),
\end{equation}
where $\zeta'(\sqrt{2})=2^{\frac{3}{4}}\mathrm{e}^{\frac{3}{4}\pi i}$ and $\zeta''(\sqrt{2})=\frac{1}{3}2^{\frac{1}{4}}\mathrm{e}^{\frac{3}{4}\pi i}$.

In order to construct the solution of $S^{(\sqrt{2})}(z)$, we introduce the $4\times 4$ matrix-valued functions for the parabolic cylinder parametrix $\Psi^{(PC)}(\zeta)$ by using the parabolic cylinder functions. Let 
\begin{equation*}
    D(\zeta)=2^{-\frac{1}{2}\hat{\sigma}_3}\begin{pmatrix}
        D_{-\nu_1-1}(i\zeta) & 0 & D_{\nu_1}(\zeta) & 0\\
        0 & D_{-\nu_2-1}(i\zeta) & 0 & D_{\nu_2}(\zeta)\\
        D'_{-\nu_1-1}(i\zeta) & 0 & D'_{\nu_1}(\zeta) & 0\\
        0 & D'_{-\nu_2-1}(i\zeta) & 0 & D'_{\nu_2}(\zeta)
    \end{pmatrix}\begin{pmatrix}
         \mathrm{e}^{\frac{\pi i}{2}(\nu_1+1)}& 0 & 0 & 0\\
        0 &  \mathrm{e}^{\frac{\pi i}{2}(\nu_2+1)} & 0 & 0\\
        0 & 0 & 1 & 0\\
        0 & 0 & 0 & 1
    \end{pmatrix},
\end{equation*}
where $D_{\nu_k}$ is the standard parabolic cylinder function with parameter $\nu_k,~k=1,2;$ see \cite{FIKY, O}. Denote
\begin{equation}
    H_0=\begin{pmatrix}
        1 & 0 & 0 & 0\\
        0 & 1 & 0 & 0\\
        h_0(\nu_1) & 0 & 1 & 0\\
        0 & h_0(\nu_2) & 0 & 1
    \end{pmatrix},
    H_1=\begin{pmatrix}
        1 & 0 & h_1(\nu_1) & 0\\
        0 & 1 & 0 & h_1(\nu_2)\\
        0 & 0 & 1 & 0\\
        0 & 0 & 0 & 1
    \end{pmatrix},
\end{equation}
\begin{equation}
     H_{k+2}=\mathrm{e}^{i\pi(\Lambda+\frac{I}{2})\otimes\sigma_3}H_k\mathrm{e}^{-i\pi(\Lambda+\frac{I}{2})\otimes\sigma_3} 
     ~,k=0,1,~
     H_4=\mathrm{e}^{2\pi i\Lambda\otimes\sigma_3},
\end{equation}
where
\begin{equation}\label{eq:def Lam}
  \Lambda=\begin{pmatrix}
        \nu_1 & 0\\
        0 & \nu_2
    \end{pmatrix}=\begin{pmatrix}
        -\frac{1}{2\pi i}\ln(1-\lambda_1^2) & 0\\
        0 &  -\frac{1}{2\pi i}\ln(1-\lambda_2^2)
    \end{pmatrix},
\end{equation}
with
\begin{equation}\label{eq:def h0,h1}
    h_0{(\nu)}=-i\frac{\sqrt{2\pi}}{\Gamma(\nu+1)}, ~h_1{(\nu)}=\frac{\sqrt{2\pi}}{\Gamma(-\nu)}\mathrm{e}^{i\pi\nu},
\end{equation}
\begin{equation}
   1+h_0{(\nu_1)}h_1{(\nu_1)}=\mathrm{e}^{2\pi i\nu_1}, ~1+h_0{(\nu_2)}h_1{(\nu_2)}=\mathrm{e}^{2\pi i\nu_2}.
\end{equation}
We define 
\begin{equation*}
    \Psi^{(PC)}(\zeta)=\begin{cases}
        D(\zeta),& \arg\zeta\in(-\frac{\pi}{4},0),\\
       D(\zeta)H_0,& \arg\zeta\in(0,\frac{\pi}{2}),\\
       D(\zeta)H_1,& \arg\zeta\in(\frac{\pi}{2},\pi),\\
       D(\zeta)H_2,& \arg\zeta\in(\pi,\frac{3\pi}{2}),\\
       D(\zeta)H_3,& \arg\zeta\in(\frac{3\pi}{2},\frac{7\pi}{4}).\\
    \end{cases}
\end{equation*}
Then $\Psi^{(PC)}(\zeta)$ satisfies 
the following RH problem.
\subsubsection*{RH problem for $\Psi^{(PC)}$}\begin{description}
    \item(1)
    $\Psi^{(PC)}(\zeta)$ is a $4\times4$ matrix-valued function analytic for $z\in\mathbb{C}\backslash\Gamma_k,~k=0,1,2,3,4.$
\item (2)
$\Psi^{(PC)}(\zeta)$ satisfies the jump condition
\begin{equation}
\Psi^{(PC)}_+(\zeta)=\Psi^{(PC)}_-(\zeta)H_k,~z\in\Gamma_k,~k=0,1,2,3,4,
\end{equation}
where $\Gamma_k,~k=0,1,...,4$ are shown in Figure \ref{fig:fig7}.
\begin{figure}[h]
    \centering
    \includegraphics[width= 0.5\textwidth]{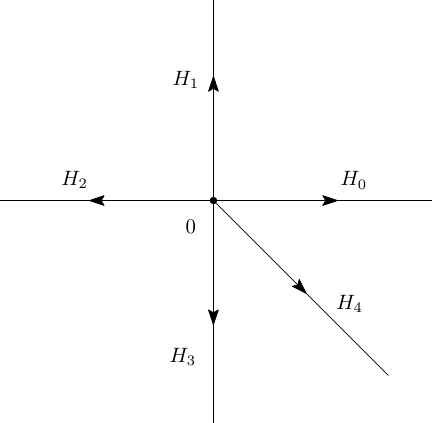}
    \caption{The jump contours and jump matrices of the RH problem for $\Psi^{(PC)}$.}
    \label{fig:fig7}
\end{figure}
\item (3)
As $\zeta\to\infty$, $\Psi^{(PC)}(\zeta)$ satisfies the following asymptotic behavior
\begin{equation}\label{eq:PC at infty}
\Psi^{(PC)}(\zeta)=\begin{pmatrix}
    0 & I\\
    I & \zeta I
\end{pmatrix}2^{\frac{1}{2}\hat{\sigma}_3}\left(I_4+\frac{G_1}{\zeta}+\frac{G_2}{\zeta^2}+\mathcal{O}(\zeta^{-3})\right)\mathrm{e}^{(\frac{\zeta^2}{4}\hat{\sigma}_3-\Lambda\otimes\sigma_3 \ln\zeta)},
\end{equation}
where
\begin{equation}\label{eq:m1 m2}
    G_1=\begin{pmatrix}
        0 & \Lambda\\
        I & 0
    \end{pmatrix},~~
    G_2=\begin{pmatrix}
        \frac{(I+\Lambda)\Lambda}{2} & 0\\
        0 & \frac{(I-\Lambda)\Lambda)}{2}
    \end{pmatrix},
\end{equation}
and $\ln\zeta$ takes the principal branch.
\end{description}
Then the solution of $S^{(\sqrt{2})}(z)$ can be constructed by using the parabolic cylinder parametrix $\Psi^{(PC)}(\zeta)$ as follows:
\begin{equation}\label{eq:local para of S}
\begin{split}
    S^{(\sqrt{2})}(z)
    =&S^{(\infty)}(z)\hat{U}(\sqrt{t}\zeta(z))^{\Lambda\otimes\sigma_3}L \mathrm{e}^{-\frac{itg(\sqrt{2})}{2}\hat{\sigma}_3}2^{-\frac{1}{2}\hat{\sigma}_3}\begin{pmatrix}
        \sqrt{t}\zeta(z) I & I\\
        I & 0
    \end{pmatrix}
    \Psi^{(PC)}(\sqrt{t}\zeta(z))\mathrm{e}^{-\frac{\zeta^2}{4}\hat{\sigma}_3}\\
    &\times\mathrm{e}^{\frac{itg(\sqrt{2})}{2}\hat{\sigma}_3}L^{-1}\hat{U}^{-1},
    \end{split}
\end{equation}
where
\begin{equation}
    \hat{U}=U\otimes I=\begin{pmatrix}
    U & 0\\
    0 & U\end{pmatrix},
\end{equation}
and
\begin{equation}\label{eq:def L1}
    L=\mathrm{diag}\left(\left(-\frac{h_1{(\nu_1)}}{\lambda_1}\right)^{-\frac{1}{2}},\left(-\frac{h_1{(\nu_2)}}{\lambda_2}\right)^{-\frac{1}{2}},\left(-\frac{h_1{(\nu_1)}}{\lambda_1}\right)^{\frac{1}{2}},\left(-\frac{h_1{(\nu_2)}}{\lambda_2}\right)^{\frac{1}{2}}\right)
   :=\begin{pmatrix}
        L_1 & 0\\
        0 & L_1^{-1}
    \end{pmatrix}.
\end{equation}
Then, from \eqref{eq:relation N and Pinfty}, \eqref{eq:solution N}, \eqref{eq:PC at infty} and \eqref{eq:local para of S}, we have 
\begin{equation}\label{eq:asySr}
    S^{(\sqrt{2})}(z)=\left(I_4+\frac{A_{\sqrt{2}}(z)G_1A_{\sqrt{2}}(z)^{-1}}{\sqrt{t}\zeta(z)}+\frac{A_{\sqrt{2}}(z)G_2A_{\sqrt{2}}(z)^{-1}}{t\zeta(z)^2
    }+\mathcal{O}(t^{-\frac{3}{2}})\right)S^{(\infty)}(z),
\end{equation}
as $t \to +\infty$ and $z \in \partial U(\sqrt{2}).$
Here we have
\begin{footnotesize}
\begin{equation}\label{eq:expression forA}
\begin{split}
    &A_{\sqrt{2}}(z)=\\
    &\begin{pmatrix}
\mathrm{e}^{\frac{2}{3}\sqrt{2}it}\mathrm{e}^{\sqrt{2}\tau i\sigma_3}{\Phi}^{(1)}(z,\tau)L_1\left(\frac{z+\sqrt{2}}{2\sqrt{2}}\right)^{\alpha\sigma_3}\left(\frac{z+\sqrt{2}}{z-\sqrt{2}}\zeta(z)\sqrt{t}\right)^{\Lambda} & 0\\
0 & \mathrm{e}^{-\frac{2}{3}\sqrt{2}it}\mathrm{e}^{\sqrt{2}\tau i\sigma_3}{\Phi}^{(0)}(-z,\tau)L_1^{-1}\left(\frac{z+\sqrt{2}}{2\sqrt{2}}\right)^{-\alpha\sigma_3}\left(\frac{z+\sqrt{2}}{z-\sqrt{2}}\zeta(z)\sqrt{t}\right)^{-\Lambda}
    \end{pmatrix},
    \end{split}
\end{equation}
\end{footnotesize}
where ${\Phi}^{(0)}(\lambda,x)$ and ${\Phi}^{(1)}(\lambda,x)$ are defined in \eqref{eq:local0} and \eqref{eq:local1}, respectively.

Near the other node point $z_-=-\sqrt{2}$, we also define that $U(-\sqrt{2})=\lbrace z: |z+\sqrt{2}|<\delta\rbrace$ with $\delta<\frac{1}{2}$. Then the parametrix $S^{(-\sqrt{2})}(z)$ can be constructed by using the  symmetry  relation
\begin{equation}\label{eq:Sl}
    S^{(-\sqrt{2})}(z)=\hat{\sigma}_1S^{(\sqrt{2})}(-z)\hat{\sigma}_1,
\end{equation}
where $\hat{\sigma}_1$ is defined after \eqref{eq:tensor}.

\subsection{Final transformation}
With all of the parametrices constructed, the final transformation can be defined as
\begin{equation}\label{eq:def R(z)}
    R(z)=\begin{cases}
        S(z)S^{(\sqrt{2})}(z)^{-1}, & z\in U(\sqrt{2})\backslash\Sigma_S,\\
        S(z)S^{(-\sqrt{2})}(z)^{-1}, & z\in U(\sqrt{-2})\backslash\Sigma_S,\\
        S(z)S^{(\infty)}(z)^{-1} & eleswhere.
    \end{cases}
\end{equation}
Then $R(z)$ satisfies the following RH problem.
\subsubsection*{RH problem for $R$}
\begin{description}
    \item(1)
    $R(z)$ is analytic for $z \in \mathbb{C}\backslash\Sigma_R$, where the contours $\Sigma_R$ are shown in Figure \ref{fig:fig8}.
    \begin{figure}[h]
    \centering
    \includegraphics[width = 0.7\textwidth ]{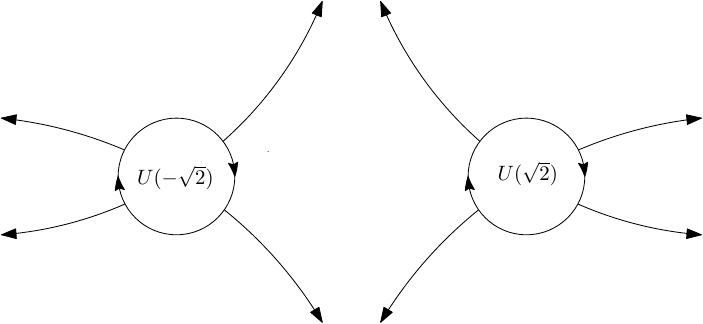}
    \caption{The contours $\Sigma_R$ of the RH problem for $R$.}
    \label{fig:fig8}
\end{figure}
\item (2)
$R(z)$ satisfies the jump condition 
\begin{equation}\label{eq:JumpR}
    R_+(z)=R_-(z)J_R(z),
\end{equation}
where
\begin{equation}\label{eq:def JR}
    J_R(z)=\begin{cases}
        S^{(\sqrt{2})}(z)S^{(\infty)}(z)^{-1}, & z \in \partial U(\sqrt{2}),\\
        S^{(-\sqrt{2})}(z)S^{(\infty)}(z)^{-1}, & z \in \partial U(-\sqrt{2}),\\
        S^{(\infty)}(z)J_S(z)S^{(\infty)}(z)^{-1} & elsewhere.
    \end{cases}
\end{equation}
\item (3)
As $z \to \infty$, we have
\begin{equation}\label{eq:asymptotic R at infty}
    R(z)=I_4+\frac{R_1(t)}{z}+\mathcal{O}(z^{-2}).
\end{equation}
\end{description}
From the matching condition \eqref{eq:matching1}, \eqref{eq:Sl} and \eqref{eq:defJs}, we have as $t
\to +\infty$,
\begin{equation}\label{eq:estimate of JR}
    J_R(z)=\begin{cases}
        I_4+\mathcal{O}\left(\frac{1}{\sqrt{t}}\right), & z \in \partial U(\pm\sqrt{2}),\\
        I_4+\mathcal{O}(\mathrm{e}^{-ct}) & eleswhere,
    \end{cases}
\end{equation}
where $c$ is a positive constant. This, together with \eqref{eq:asySr}, \eqref{eq:Sl} and \eqref{eq:def JR}, implies that $J_R(z)$ has an expansion of the form 
\begin{equation}\label{eq:asyJR}
    J_R(z)=I_4+\frac{J_{R,1}(z)}{t^{\frac{1}{2}}}+\frac{J_{R,2}(z)}{t}+\mathcal{O}\left(\frac{1}{t^{\frac{3}{2}}}\right),~~~t\to+\infty,
\end{equation}
where
\begin{equation}\label{eq:formaula of JR}
    J_{R,j}(z)=\begin{cases}
        \frac{A_{\sqrt{2}}(z)G_jA_{\sqrt{2}}(z)^{-1}}{\zeta(z)^j},&z\in\partial U(\sqrt{2})\\
        \frac{\hat{\sigma}_1A_{\sqrt{2}}(-z)G_jA_{\sqrt{2}}(-z)^{-1}\hat{\sigma}_1}{\zeta(-z)^j},&z\in\partial U(-\sqrt{2})
    \end{cases},j=1,2,
\end{equation}
with $A_{\sqrt{2}}(z)$ given in \eqref{eq:expression forA}. We proceed to calculate the Laurent expansion on $J_{R,1}(z)$ and $J_{R,2}(z)$ at $z=\sqrt{2}$ for later use.
\begin{equation}\label{eq:laurent ex JR1}
    J_{R,1}(z)=\frac{C_{-1}^{(1)}}{z-\sqrt{2}}+C_0^{(1)}+\mathcal{O}(z-\sqrt{2}),~~~z\to \sqrt{2},
\end{equation}
\begin{equation}\label{eq:laurent ex JR2}
    J_{R,2}(z)=\frac{C_{-2}^{(2)}}{(z-\sqrt{2})^2}+\frac{C_{-1}^{(2)}}{z-\sqrt{2}}+C_0^{(2)}+\mathcal{O}(z-\sqrt{2}),~~~z\to \sqrt{2},
\end{equation}
where
\begin{equation}\label{eq:C-1,12}
\left(C_{-1}^{(1)}\right)_{12}=\mathrm{e}^{\frac{4}{3}\sqrt{2}it}\mathrm{e}^{\sqrt{2}\tau i\sigma_3}{\Phi}^{(1)}_0(x)t^{\Lambda}L_1^2\Lambda(\zeta{'}(\sqrt{2}))^{2\Lambda-1}(2\sqrt{2})^{2\Lambda}{\Phi}^{(0)}_0(x)^{-1}\mathrm{e}^{-\sqrt{2}\tau i\sigma_3},     \end{equation}
\begin{equation}\label{eq:C-1,21}
     \left(C_{-1}^{(1)}\right)_{21}=\mathrm{e}^{-\frac{4}{3}\sqrt{2}it}\mathrm{e}^{\sqrt{2}\tau i\sigma_3}{\Phi}^{(1)}_0(x)t^{-\Lambda}L_1^{-2}(\zeta{'}(\sqrt{2}))^{-2\Lambda-1}(2\sqrt{2})^{-2\Lambda}{\Phi}^{(1)}_0(x)^{-1}\mathrm{e}^{-\sqrt{2}\tau i\sigma_3}, 
\end{equation}
and
\begin{equation}
    \begin{split}
&\left(C_0^{(1)}\right)_{12}=\mathrm{e}^{\frac{4}{3}\sqrt{2}it}\mathrm{e}^{\sqrt{2}\tau i\sigma_3} \Bigg(\frac{\alpha}{\sqrt{2}}{\Phi}^{(1)}_0(x)\sigma_3t^{\Lambda}L_1^2\Lambda(\zeta'(\sqrt{2}))^{2\Lambda-1} (2\sqrt{2})^{2\Lambda}{\Phi}^{(0)}_0(x)^{-1}+{\Phi}^{(1)}_0(x)t^{\Lambda}L_1^2\Lambda\\& (\zeta'(\sqrt{2}))^{2\Lambda-1}(2\sqrt{2})^{2\Lambda}\left(\frac{\Lambda}{\sqrt{2}}+\frac{2\Lambda-I}{2\zeta'(\sqrt{2})}\zeta''(\sqrt{2})\right){\Phi}^{(0)}_0(x)^{-1}+\frac{1}{2\sqrt{2}}{\Phi}^{(1)}_0(x){\Phi}^{(1)}_1(x)t^{\Lambda}L_1^2\Lambda(\zeta'(\sqrt{2}))^{2\Lambda-1}\\&(2\sqrt{2})^{2\Lambda}{\Phi}^{(0)}_0(x)^{-1}+\frac{1}{2\sqrt{2}}{\Phi}^{(1)}_0(x)t^{\Lambda}L_1^2\Lambda(\zeta'(\sqrt{2}))^{2\Lambda-1}(2\sqrt{2})^{2\Lambda}{\Phi}^{(0)}_1(x){\Phi}^{(0)}_0(x)^{-1}\Bigg)\mathrm{e}^{-\sqrt{2}\tau i\sigma_3},
\end{split}
\end{equation}
\begin{equation}
    \begin{split}
                & \left(C_0^{(1)}\right)_{21}=-\mathrm{e}^{-\frac{4}{3}\sqrt{2}it}\mathrm{e}^{\sqrt{2}\tau i\sigma_3}\Bigg(\frac{\alpha}{\sqrt{2}}{\Phi}^{(0)}_0(x)\sigma_3t^{-\Lambda}L_1^{-2}(\zeta'(\sqrt{2}))^{-2\Lambda-1} (2\sqrt{2})^{-2\Lambda}{\Phi}^{(1)}_0(x)^{-1}{\Phi}^{(0)}_0(x)t^{-\Lambda}L_1^{-2}\\& (\zeta'(\sqrt{2}))^{-2\Lambda-1}(2\sqrt{2})^{-2\Lambda}\left(\frac{\Lambda}{\sqrt{2}}+\frac{2\Lambda-I}{2\zeta'(\sqrt{2})}\zeta''(\sqrt{2})\right){\Phi}^{(1)}_0(x)^{-1}+\frac{1}{2\sqrt{2}}{\Phi}^{(0)}_0(x){\Phi}^{(0)}_1(x)t^{-\Lambda}L_1^{-2}(\zeta'(\sqrt{2}))^{-2\Lambda-1}\\&(2\sqrt{2})^{-2\Lambda}{\Phi}^{(1)}_0(x)^{-1}+\frac{1}{2\sqrt{2}}{\Phi}^{(0)}_0(x)t^{-\Lambda}L_1^{-2}(\zeta'(\sqrt{2}))^{-2\Lambda-1}(2\sqrt{2})^{-2\Lambda}{\Phi}^{(1)}_1(x){\Phi}^{(1)}_0(x)^{-1}\Bigg)\mathrm{e}^{-\sqrt{2}\tau i\sigma_3},
    \end{split}
\end{equation}
with $L_1$, $\Lambda$ given in \eqref{eq:def Lam} and \eqref{eq:def L1}, respectively, and 
\begin{equation}\label{eq:resJR2_11}
    \begin{split}
  \left(C_{-1}^{(2)}\right)_{11}=&\mathrm{e}^{\sqrt{2}\tau i\sigma_3}\bigg(
  -\frac{\zeta{''}(\sqrt{2})}{2(\zeta{'}(\sqrt{2}))^3}{\Phi}^{(1)}_0(x)(I+\Lambda)\Lambda{\Phi}^{(1)}_0(x)^{-1}
  +\frac{1}{4\sqrt{2}(\zeta'(\sqrt{2}))^2}\Big({\Phi}^{(1)}_0(x){\Phi}^{(1)}_1(x)\\&\times(I+\Lambda)\Lambda{\Phi}^{(1)}_0(x)^{-1}
  -{\Phi}^{(1)}_0(x)(I+\Lambda)\Lambda{\Phi}^{(1)}_1(x){\Phi}^{(1)}_0(x)^{-1}\Big)\bigg)\mathrm{e}^{-\sqrt{2}\tau i\sigma_3},
\end{split}
\end{equation}
\begin{equation}\label{eq:resJR2_22}
    \begin{split}
  \left(C_{-1}^{(2)}\right)_{22}=&\mathrm{e}^{\sqrt{2}\tau i\sigma_3}\bigg(
  -\frac{\zeta{''}(\sqrt{2})}{2(\zeta{'}(\sqrt{2}))^3}{\Phi}^{(0)}_0(x)(I-\Lambda)\Lambda{\Phi}^{(0)}_0(x)^{-1}-\frac{1}{4\sqrt{2}(\zeta'(\sqrt{2}))^2}\Big({\Phi}^{(0)}_0(x){\Phi}^{(0)}_1(x)\\&\times(I-\Lambda)\Lambda{\Phi}^{(0)}_0(x)^{-1}
  -{\Phi}^{(0)}_0(x)(I-\Lambda)\Lambda{\Phi}^{(0)}_1(x){\Phi}^{(0)}_0(x)^{-1}\Big)\bigg)\mathrm{e}^{-\sqrt{2}\tau i\sigma_3},
\end{split}
\end{equation}
with $\Phi_0^{(0)}(x)$, $\Phi_1^{(0)}(x)$ and $\Phi_0^{(1)}(x)$,
$\Phi_1^{(1)}(x)$ defined in \eqref{eq:expansion Phi 0} and \eqref{eq:expansion Phi 1}, respectively. 
Furthermore, to simplify the subsequent calculations, we define
\begin{equation}\label{eq:def m1,m2}
\begin{split}
    \mathrm{e}^{\frac{4}{3}\sqrt{2}it}t^{\Lambda}L_1^2\Lambda(\zeta{'}(\sqrt{2}))^{2\Lambda}&(2\sqrt{2})^{2\Lambda}:=\begin{pmatrix}
        m_1 & 0\\
        0 & m_2
    \end{pmatrix}\\=&\mathrm{e}^{\frac{4}{3}\sqrt{2}it}\begin{pmatrix}
        -\frac{\lambda_1}{h_1{(\nu_1)}}\left(2\sqrt{2}\mathrm{e}^{\frac{3}{4}\pi i}\sqrt{t}2^{\frac{3}{4}}\right)^{2\nu_1}\nu_1 & 0\\
        0 &  -\frac{\lambda_2}{h_1{(\nu_2)}}\left(2\sqrt{2}\mathrm{e}^{\frac{3}{4}\pi i}\sqrt{t}2^{\frac{3}{4}}\right)^{2\nu_2}\nu_2 
    \end{pmatrix}.
    \end{split}
\end{equation}
Note that $\nu_j=-\frac{1}{2\pi i}\ln(1-\lambda_j^2)\in i\mathbb{R},~j=1,2$, it then follows from the reflection formula of Gamma function that
\begin{equation}
    |\Gamma(\nu_j)|^2=\Gamma(\nu_i)\Gamma(-\nu_j)=-\frac{\pi}{\nu_j\sin(\nu_j\pi)}=\frac{2\pi}{i\nu_j\lambda_j^2}\mathrm{e}^{-i\pi\nu_j},~~j=1,2.
\end{equation}
By further substituting \eqref{eq:def h0,h1} into \eqref{eq:def m1,m2}, we obtain that
\begin{equation}\label{eq:formula m1,m2}
    m_j=\frac{\sqrt{2\pi}}{\lambda_j|\Gamma(\nu_j)|}\mathrm{e}^{-\frac{\pi}{2}i\nu_j+\frac{4}{3}\sqrt{2}it-i\arg\Gamma(\nu_j)+\frac{\pi}{2}i+\nu_j\ln(2^{\frac{9}{2}}t)},~~j=1,2.
\end{equation}
Similarly, we have
\begin{equation}
\begin{split}
    \mathrm{e}^{-\frac{4}{3}\sqrt{2}it}t^{-\Lambda}L_1^{-2}(\zeta{'}(\sqrt{2}))^{-2\Lambda}&(2\sqrt{2})^{-2\Lambda}:=\begin{pmatrix}
        n_1 & 0\\
        0 & n_2
    \end{pmatrix},
    \end{split}
    \end{equation}
where 
\begin{equation}\label{eq:def n1,n2}
    n_j=\frac{\sqrt{2\pi}}{\lambda_j|\Gamma(\nu_j)|}\mathrm{e}^{-\frac{\pi}{2}i\nu_j-\frac{4}{3}\sqrt{2}it+i\arg\Gamma(\nu_j)+\pi i-\nu_j\ln(2^{\frac{9}{2}}t)},~~j=1,2.
\end{equation}
It can be seen that $m_j$ and $n_j$ satisfy the following relation
\begin{equation}\label{eq:relation m and n}
    im_j=\overline{n_j},~~~j=1,2.
\end{equation}
By a standard theory, cf. \cite{FIKY}, the estimate \eqref{eq:estimate of JR} ensures unique solvability of the RH problem for $R(z)$ as $t\to+\infty$. Furthermore, as $t\to+\infty$, we have
\begin{equation}\label{eq:asyR}
    R(z)=I_4+\frac{R^{(1)}(z)}{t^{\frac{1}{2}}}+\frac{R^{(2)}(z)}{t}+\mathcal{O}\left(\frac{1}{t^{\frac{3}{2}}}\right),
\end{equation}
where the error term is uniform for $z \in \mathbb{C}\backslash\Sigma_R$.

We conclude this subsection with the calculations of the functions $R^{(1)}(z)$ and $R^{(2)}(z)$ in \eqref{eq:asyR} for later use. Inserting \eqref{eq:asyJR}, \eqref{eq:asyR} into the jump condition \eqref{eq:JumpR} for $R(z)$ yields
\begin{equation}\label{eq:jumpR1}
    R^{(1)}_+(z)=R^{(1)}_-(z)+J_{R,1}(z),
\end{equation}
\begin{equation}\label{eq:jumpR2}
    R^{(2)}_+(z)=R^{(2)}_-(z)+R^{(1)}_-(z)J_{R,1}(z)+J_{R,2}(z),
\end{equation}
for $z\in\partial U(\sqrt{2})\cup\partial U(-\sqrt{2})$. This, together with the facts that $R^{(1)}(z)=\mathcal{O}(z^{-1})$ and $R^{(2)}(z)=\mathcal{O}(z^{-1})$ as $z\to\infty$, we have the integral representation of $R^{(1)}(z)$ and $R^{(2)}(z)$ by using Plemelj's formula
\begin{equation}\label{eq:intexpR1}
    R^{(1)}(z)=\frac{1}{2\pi i}\int_{\partial U(\sqrt{2})\cup\partial U(-\sqrt{2})}\frac{J_{R,1}(\zeta)}{\zeta-z}d\zeta,
\end{equation}
\begin{equation}\label{eq:intexpR2}
    R^{(2)}(z)=\frac{1}{2\pi i}\int_{\partial U(\sqrt{2})\cup\partial U(-\sqrt{2})}\frac{J_{R,2}(\zeta)+R^{(1)}_-(\zeta)J_{R,1}(\zeta)}{\zeta-z}d\zeta.
\end{equation}

From the definition of $J_R(z)$ given in \eqref{eq:asyJR}, we obtain the following expressions for $R^{(1)}(z)$ and $R^{(2)}(z)$ by using residue theorem 
\begin{equation}\label{eq:expR1}
R^{(1)}(z)=\frac{A_1}{z-\sqrt{2}}-\frac{\hat{\sigma}_1A_1\hat{\sigma}_1}{z+\sqrt{2}}, ~z\in\mathbb{C}\backslash(U(\sqrt{2})\cup U(-\sqrt{2})),
\end{equation}
where
\begin{equation}\label{eq:expA1}
A_1=C_{-1}^{(1)}.
\end{equation}
To calculate $R^{(2)}(z)$, it is readily seen from \eqref{eq:def JR} and \eqref{eq:intexpR2} that
\begin{equation}\label{eq:expR2}
    R^{(2)}(z)=\frac{A_2}{z-\sqrt{2}}-\frac{\hat{\sigma}_1A_2\hat{\sigma}_1}{z+\sqrt{2}}+\frac{A_3}{(z-\sqrt{2})^2}+\frac{\hat{\sigma}_1A_3\hat{\sigma}_1}{(z+\sqrt{2})^2},~ z\in\mathbb{C}\backslash\left(U(\sqrt{2}) \cup U(-\sqrt{2})\right),
\end{equation}
where
\begin{equation}\label{eq:expA2}
    A_2=R_-^{(1)}(\sqrt{2})A_1+C_{-1}^{(2)}.
\end{equation}
In view of \eqref{eq:intexpR1}, 
\eqref{eq:expR1} and \eqref{eq:laurent ex JR1}, it follows that
\begin{equation}\label{eq:expR1-}
    R_-^{(1)}(\sqrt{2})=-\frac{\hat{\sigma}_1A_1\hat{\sigma}_1}{2\sqrt{2}}-C_0^{(1)}.
\end{equation}
From \eqref{eq:expA2}, \eqref{eq:expA1}
and \eqref{eq:expR1-}, we conclude that $A_2$ is a block diagonal matrix of the form
\begin{equation}\label{eq:result of A2}
\begin{small}
   A_2=\begin{pmatrix}
       -\frac{1}{2\sqrt{2}}\left(C_{-1}^{(1)}\right)_{21}^2-\left(C_0^{(1)}\right)_{12}\left(C_{-1}^{(1)}\right)_{21}+\left(C_{-1}^{(2)}\right)_{11} & 0\\
       0 &  -\frac{1}{2\sqrt{2}}\left(C_{-1}^{(1)}\right)_{12}^2-\left(C_0^{(1)}\right)_{21}\left(C_{-1}^{(1)}\right)_{12}+\left(C_{-1}^{(2)}\right)_{22} 
   \end{pmatrix}.
   \end{small}
\end{equation}
From the expression of the logarithmic derivative of the Fredholm determinant given in \eqref{eq:RH for FD}, to obtain its explicit formula, we only need to compute the trace of the components $(A_2)_{22}$ and $(A_2)_{11}$. From \eqref{eq:C-1,12}-\eqref{eq:resJR2_22} and \eqref{eq:result of A2}, we  derive the following result after some direct calculations
\begin{equation}\label{eq:traceA2}
\begin{split}
    &Tr(A_2-\hat{\sigma}_1A_2\hat{\sigma}_1))_{22}=\\
    &-\frac{1}{2\sqrt{2}}Tr\left((C_{-1}^{(1)})_{12}^2-(C_{-1}^{(1)})_{21}^2\right)+\frac{1}{\sqrt{2}\zeta'(\sqrt{2})^2}Tr\left(3\Lambda^2+{\Phi}^{(1)}_1(x)\Lambda+{\Phi}^{(0)}_1(x)\Lambda+2\alpha\sigma_3\Lambda\right),
    \end{split}
\end{equation}
where ${\Phi}^{(0)}_1(x)$ and
${\Phi}^{(1)}_1(x)$ are defined in \eqref{eq:expansion Phi 0} and \eqref{eq:expansion Phi 1}, and 
$\Lambda$, $\left(C_{-1}^{(1)}\right)_{12}$ and $\left(C_{-1}^{(1)}\right)_{21}$ are defined in \eqref{eq:def Lam}, \eqref{eq:C-1,12} and \eqref{eq:C-1,21}, respectively.

\section{Proof of Theorems \ref{thm:4} and \ref{thm5}}\label{sec:proof 4 and 5}
\subsection{Proof of Theorem \ref{thm5}}
Tracking back the series of transformations performed in Section \ref{sec: large t}
\begin{equation*}
    \Psi \mapsto Y \mapsto S \mapsto R,
\end{equation*}
we have, for large $z$,
\begin{equation}\label{eq:defPsi}
    \Psi(z)=R(z)S^{(\infty)}(z)\mathrm{e}^{-\frac{itg(z)}{2}\sigma_3\otimes\sigma_3},
\end{equation}
where $g(z)$, $S^{(\infty)}(z)$ and $R(z)$ are defined in \eqref{eq:def g}, \eqref{eq:relation N and Pinfty} and \eqref{eq:def R(z)}.
Then, combining \eqref{eq:defPsi}, \eqref{eq:asyPsi} and \eqref{eq:asyPinfty}, we find that 
\begin{equation}
    \Psi_1(t)\mathrm{e}^{\Theta(\lambda)}=R_1(t)+S^{(\infty)}_1.
\end{equation}
Consequently, we have
\begin{equation}\label{eq:formula alpha1}
    {H}_{II}(\Vec{s})=\sqrt{-s}\left((R_1(t))_{22}+(S_1^{(\infty)})_{22}\right),
\end{equation}
\begin{equation}\label{eq:formula beta}
     \beta(\Vec{s})=-i\sqrt{-s}\left((R_1(t))_{21}+(S_1^{(\infty)})_{21}\right),
\end{equation}
with $t=(-s)^{\frac{3}{2}}$. A combination of \eqref{eq:asymptotic R at infty}, \eqref{eq:asyR}, \eqref{eq:expR1} and \eqref{eq:expR2} shows that 
\begin{equation}\label{eq:def R1(t)}
    R_1(t)=\frac{A_1-\hat{\sigma}_1A_1\hat{\sigma}_1}{\sqrt{t}}+\frac{A_2-\hat{\sigma}_1A_2\hat{\sigma}_1}{t}+\mathcal{O}(t^{-\frac{3}{2}}),
\end{equation}
with $A_1$ and $A_2$ given in \eqref{eq:expA1} and \eqref{eq:expA2}.

From \eqref{eq:asyofPVinfty}, \eqref{eq:solution of f}, \eqref{eq:solution of N1} and \eqref{eq:relation N and Pinfty}, we  obtain
\begin{equation}\label{eq:P1inf22}
   (S_1^{(\infty)})_{22}=2\sqrt{2}\mathrm{e}^{\sqrt{2}\tau i\sigma_3} 
         \left(-\Phi_{-1}(4\sqrt{2}\tau i)+\frac{1}{4\pi i}\ln[(1-\lambda_1^2)(1-\lambda_2^2)]I\right)\mathrm{e}^{-\sqrt{2}\tau i\sigma_3},
\end{equation}
\begin{equation}\label{eq:S21=0}
    (S_1^{(\infty)})_{21}=0,
\end{equation}
where $\Phi_{-1}(x)$ is defined in \eqref{eq:RH Ham for PV}.
From \eqref{eq:expA1}, \eqref{eq:result of A2}, \eqref{eq:def R1(t)} and the fact that $t=(-s)^{\frac{3}{2}}$, we have, as $s\to-\infty$, 
\begin{equation}\label{eq:R1(t)22}
    (R_1(t))_{22}=\mathcal{O}((-s)^{-\frac{3}{2}}),
\end{equation}
\begin{equation}\label{eq:R1(t)21}
    (R_1(t))_{21}=(-s)^{-\frac{3}{4}}\Big(\left(C_{-1}^{(1)}\right)_{21}-\left(C_{-1}^{(1)}\right)_{12}\Big)+\mathcal{O}((-s)^{-\frac{3}{2}}),
\end{equation}
Combining \eqref{eq:formula alpha1}, \eqref{eq:P1inf22} and \eqref{eq:R1(t)22}, we obtain the asymptotic expansion of $H_{II}(\Vec{s})$ \eqref{eq:asy Ham for NC PII} as $s\to-\infty$.

Similarly, from \eqref{eq:formula beta}, \eqref{eq:S21=0} and \eqref{eq:R1(t)21}, we have, as $s\to-\infty$,
\begin{equation}
    \beta(\Vec{s})=-i(-s)^{-\frac{1}{4}}\Big(\left(C_{-1}^{(1)}\right)_{21}-\left(C_{-1}^{(1)}\right)_{12}\Big)+\mathcal{O}((-s)^{-1}).
\end{equation}
This, together with \eqref{eq:C-1,12}, \eqref{eq:C-1,21}, \eqref{eq:def m1,m2}, \eqref{eq:formula m1,m2}-\eqref{eq:relation m and n} and \eqref{eq:analyticfac},
we finally obtain \eqref{eq:asy NC PII}. This completes the proof of Theorem \ref{thm5}. 

\subsection{Proof of Theorem \ref{thm:4}}\label{sec:proof FD}
From  \eqref{eq:RH for FD},  \eqref{eq:formula alpha1}, \eqref{eq:def R1(t)} and \eqref{eq:P1inf22}, we have as $s\to-\infty$
\begin{equation}\label{eq:asy alpha}
  \partial_s\ln\det(Id-Ai^2_{\Vec{s}})=-\frac{2\sqrt{2}}{\pi}\ln(1-\lambda_1^2)(1-\lambda_2^2)\sqrt{-s}+\frac{2i}{s}Tr(A_2-\hat{\sigma}_1A_2\hat{\sigma}_1)_{22}+\mathcal{O}((-s)^{-\frac{7}{4}}).
\end{equation}

We proceed to calculate the second term on the right hand-side of the above equation using \eqref{eq:traceA2}.
We start by calculate the first term in  \eqref{eq:traceA2}.
Combining \eqref{eq:analyticfac}, \eqref{eq:C-1,12},  \eqref{eq:def m1,m2},
\eqref{eq:formula m1,m2} and applying the symmetry relations \eqref{eq:symmetry 0} and \eqref{eq:symmetry 1}, we have 
\begin{equation}\label{eq:trA_12^2}
    Tr\left((C_{-1}^{(1)})_{12}^2\right)=\frac{1}{\zeta'(\sqrt{2})^2}\left(\left(m_1\overline{\hat{u}(4\sqrt{2}\tau i)}+m_2\hat{u}(4\sqrt{2}\tau i)\right)^2-2m_1m_2\right),
\end{equation}
By a similar calculation, we find that
\begin{equation}\label{eq:trA_21^2}
    Tr\left((C_{-1}^{(1)})_{21}^2\right)=\frac{1}{\zeta'(\sqrt{2})^2}\left(\left(n_1\hat{u}(4\sqrt{2}\tau i)+n_2\overline{\hat{u}(4\sqrt{2}\tau i)}\right)^2-2n_1n_2\right),
\end{equation}
where    $\zeta'(\sqrt{2})$ is given in \eqref{eq:expansion for comformal}, $\hat{u}(x)$ is defined in \eqref{eq:def hatu}, and $m_j,~n_j,~j=1,2$ are defined in \eqref{eq:formula m1,m2} and \eqref{eq:def n1,n2}, respectively. 
This, together with \eqref{eq:relation m and n}, implies the following expression for the first term in  \eqref{eq:traceA2}
\begin{equation}\label{eq:TrA12-A21}
\begin{split}
    &Tr\left((C_{-1}^{(1)})_{12}^2-(C_{-1}^{(1)})_{21}^2\right)=\frac{1}{\zeta'(\sqrt{2})^2}\Big((m_1\overline{\hat{u}(4\sqrt{2}\tau i)})^2+(\overline{m_1}\hat{u}(4\sqrt{2}\tau i))^2+((m_2\hat{u}(4\sqrt{2}\tau i))^2\\
    &+(\overline{m_2\hat{u}(4\sqrt{2}\tau i)})^2+2(m_1m_2+\overline{m_1m_2})(|\hat{u}(4
    \sqrt{2}\tau i)|^2-1)\Big).
    \end{split}
\end{equation}

Next, we  calculate the  expression for the remaining term in  \eqref{eq:traceA2}.
Using \eqref{eq:trace0}, \eqref{eq:trace1} and combining the definition of the Hamiltonian $H_V(x)$ \eqref{eq:def for HV}, we obtain 
\begin{equation}\label{eq:trace2}
\begin{split}
     Tr\left({\Phi}^{(0)}_1(x)\Lambda\right)&=Tr\bigg({\Phi}^{(1)}_1(x)\Lambda\bigg)\\
     &=-\frac{\nu_1-\nu_2}{2\alpha}\left((-\frac{x}{2}+v-\alpha)(2v-2\alpha)-v(v-2\alpha)(u+\frac{1}{u})\right) \\&=-\frac{\nu_1-\nu_2}{2\alpha}(xH_V(x)+2\alpha^2),
     \end{split}
\end{equation}
with $v(x)$ given in \eqref{eq:PV for v}, and $x=4\sqrt{2}\tau i$.

By inserting   \eqref{eq:traceA2}, \eqref{eq:TrA12-A21} and  \eqref{eq:trace2}  into  \eqref{eq:asy alpha}, we finally obtain \eqref{eq:asy of FD}. This completes the proof of Theorem \ref{thm:4}.

\section*{Acknowledgements} 
We  thank the anonymous referees for their helpful comments and constructive suggestions. 
The work of Shuai-Xia Xu was supported in part by the National Natural Science Foundation of China under grant numbers  12431008, 12371257 and 11971492, and by  Guangdong Basic and Applied Basic Research Foundation (Grant No. 2022B1515020063). Yu-Qiu Zhao was supported in part by the National Natural Science Foundation of China under grant numbers  11971489 and 12371077. 
\appendix
\section{Confluent hypergeometric parametrix}\label{sec:Appendix}
The confluent hypergeometric parametrix $\Phi^{(CHF)}(\zeta)$ is a solution of the following RH problem; see \cite{IK}.
\subsection*{RH problem for $\Phi^{(CHF)}$}
\begin{description}
    \item(1)
    $\Phi^{(CHF)}(\zeta)=\Phi^{(CHF)}(\zeta,a,b)$ is analytic for $\zeta\in\mathbb{C}\backslash\lbrace\cup_{k=1}^{8}\Gamma_k\rbrace$, where the contours are defined below
\begin{equation*}
    \Gamma_1=\mathrm{e}^{\frac{\pi i}{2}}\mathbb{R}^+,~ \Gamma_2=\mathrm{e}^{\frac{3\pi i}{4}}\mathbb{R}^+, ~\Gamma_3=(-\infty,0),~ \Gamma_4=-\mathrm{e}^{\frac{-3\pi i}{4}}\mathbb{R}^+, 
\end{equation*}
\begin{equation*}
    \Gamma_5=-\mathrm{e}^{\frac{-\pi i}{2}}\mathbb{R}^+,~ \Gamma_6=-\mathrm{e}^{\frac{-\pi i}{4}}\mathbb{R}^+,~ \Gamma_7=(0,\infty), ~\Gamma_8=\mathrm{e}^{\frac{\pi i}{4}}\mathbb{R}^+;
\end{equation*}
see Figure \ref{fig:fig9} for an illustration.
\item (2)
$\Phi^{(CHF)}(\zeta)$ satisfies the jump condition
\begin{equation}\label{eq:jump CHF}
    \Phi^{(CHF)}_+(\zeta)=\Phi^{(CHF)}_-(\zeta)J_k, ~z\in\Gamma_k,~ k=1,2,...,8,
\end{equation}
where
\begin{equation*}
    J_1=\begin{pmatrix}
        0 & \mathrm{e}^{-\pi ib}\\
        -\mathrm{e}^{\pi ib} & 0
    \end{pmatrix},
     J_2=\begin{pmatrix}
        1 & 0\\
        \mathrm{e}^{\pi i(b-2a)} & 1
    \end{pmatrix},
     J_3=\begin{pmatrix}
        \mathrm{e}^{\pi ia} & 0\\
        0 & \mathrm{e}^{-\pi ia}
    \end{pmatrix},
    J_4=\begin{pmatrix}
        1 & 0\\
        \mathrm{e}^{\pi i(2a-b)} & 1
    \end{pmatrix},
\end{equation*}
\begin{equation*}
    J_5=\begin{pmatrix}
        0 & \mathrm{e}^{\pi ib}\\
        -\mathrm{e}^{-\pi ib} & 0
    \end{pmatrix},
     J_6=\begin{pmatrix}
        1 & 0\\
        \mathrm{e}^{-\pi i(b+2a)} & 1
    \end{pmatrix},
    J_7=\begin{pmatrix}
        \mathrm{e}^{\pi ia} & 0\\
        0 & \mathrm{e}^{-\pi ia}
    \end{pmatrix},
     J_8=\begin{pmatrix}
        1 & 0\\
        \mathrm{e}^{\pi i(b+2a)} & 1
    \end{pmatrix}.
\end{equation*}
\item (3)
As $\zeta\to\infty$, we have 
\begin{equation}
   \Phi^{(CHF)}(\zeta)=\left(I+\frac{a^2-b^2}{\zeta}\begin{pmatrix}
       1 & \frac{\Gamma(a-b)}{\Gamma(a+b+1)}\\
       -\frac{\Gamma(a+b)}{\Gamma(a-b+1)} & -1
   \end{pmatrix}   
   +\mathcal{O}(\zeta^{-2})\right)\zeta^{-b\sigma_3}\mathrm{e}^{-\frac{1}{2}\zeta\sigma_3}C(\zeta),
\end{equation}
where the branch cut for $\zeta^{-b}$ is taken along the negative imaginary axis such that $\arg \zeta\in
\left(-\frac{\pi}{2},\frac{3\pi}{2}\right)$. And $C(\zeta)$ is the following piece-wise constant matrix 
\begin{equation}
    C(\zeta)=\begin{cases}
        \mathrm{e}^{-\frac{1}{2}\pi ia\sigma_3}\mathrm{e}^{\pi ib\sigma_3}, & \frac{\pi}{2} <\arg\zeta <\pi,\\
        \mathrm{e}^{\frac{1}{2}\pi ia\sigma_3}\mathrm{e}^{\pi ib\sigma_3}, & \pi <\arg\zeta <\frac{3\pi}{2},\\
        \begin{pmatrix}
            0 & -1\\
            1 & 0
        \end{pmatrix}\mathrm{e}^{-\frac{1}{2}\pi ia\sigma_3}, & -\frac{\pi}{2} <\arg\zeta <0,\\
        \begin{pmatrix}
            0 & -1\\
            1 & 0
        \end{pmatrix}\mathrm{e}^{\frac{1}{2}\pi ia\sigma_3}, & 0 < \arg\zeta <\frac{\pi}{2}. 
    \end{cases}
\end{equation}
\begin{figure}
    \centering
    \includegraphics[width = 0.45\textwidth]{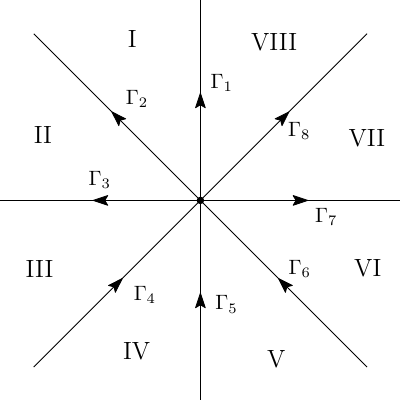}
    \caption{The contours of the RH problem for $\Phi^{(CHF)}$.}
    \label{fig:fig9}
\end{figure}
\item (4)
As $\zeta\to 0,~ \Phi^{(CHF)}(\zeta)$ satisfies the following asymptotic behavior
\begin{equation*}
    \Phi^{(CHF)}(\zeta)=\begin{cases}
        \mathcal{O}\begin{pmatrix}
            \zeta^a & \zeta^{-a}\\
            \zeta^a & \zeta^{-a}
        \end{pmatrix}, & \zeta\in II\cup\ III\cup VI\cup VII,\\
        \mathcal{O}\begin{pmatrix}
            \zeta^a & \zeta^{a}\\
            \zeta^a & \zeta^{a}
        \end{pmatrix}, & \zeta\in I\cup IV\cup V\cup VIII,
    \end{cases}
\end{equation*}
for $a > 0$, 
\begin{equation*}
     \Phi^{(CHF)}(\zeta)=\begin{cases}
        \mathcal{O}\begin{pmatrix}
            1 & \ln\zeta\\
            1 & \ln\zeta
        \end{pmatrix}, & \zeta\in II\cup III\cup VI\cup VII,\\
        \mathcal{O}\begin{pmatrix}
            \ln\zeta & \ln\zeta\\
            \ln\zeta & \ln\zeta
        \end{pmatrix}, & \zeta\in I\cup IV\cup V\cup VIII,
    \end{cases}
\end{equation*}
for $a = 0$, and for $a < 0$
\begin{equation*}
    \Phi^{(CHF)}(\zeta)=\mathcal{O}\begin{pmatrix}
        \zeta^a & \zeta^a\\
        \zeta^a & \zeta^a
    \end{pmatrix}, ~\zeta\in\mathbb{C}\backslash\lbrace\cup_{k=1}^{8}\Gamma_k\rbrace.
\end{equation*}
\end{description}
Let $G$ and $H$ be functions defined in terms of the standard Whittaker function $M(\zeta)$ and $W(\zeta)$ \cite{O} as follows:
\begin{equation}
    G(a,b,\zeta)=\zeta^{-\frac{1}{2}}M_{\frac{1}{2}+\frac{b}{2}-a,\frac{b}{2}}(\zeta),~~~~H(a,b,\zeta)=\zeta^{-\frac{1}{2}}W_{\frac{1}{2}+\frac{b}{2}-a,\frac{b}{2}}(\zeta).
\end{equation}
Then the solution for the above RH problem can be given explicitly, for $\zeta\in II$, we have
\begin{equation}
    \Phi^{(CHF)}_0(\zeta)=\begin{pmatrix}
        \frac{\Gamma(1+a-b)}{\Gamma(1+2a)}G(a+b,2a,\zeta)\mathrm{e}^{-\frac{3}{2}\pi ia} & -\frac{\Gamma(1+a-b)}{\Gamma(a+b)}H(1+a-b,2a,\mathrm{e}^{-\pi i}\zeta)\mathrm{e}^{\frac{1}{2}\pi ia}\\
            \frac{\Gamma(1+a+b)}{\Gamma(1+2a)}G(1+a+b,2a,\zeta)\mathrm{e}^{-\frac{3}{2}\pi ia} & H(a-b,2a,\mathrm{e}^{-\pi i}\zeta)\mathrm{e}^{\frac{1}{2}\pi ia}
    \end{pmatrix}.
\end{equation}
The solution in other sectors can be obtained from the above formula and the jump condition in \eqref{eq:jump CHF}.

\end{document}